\documentclass[preprint,9pt,numafflabel]{elsarticle}
\makeatletter
\def\ps@pprintTitle{%
 \let\@oddhead\@empty
 \let\@evenhead\@empty
 \def\@oddfoot{\centerline{\thepage}}%
 \let\@evenfoot\@oddfoot}
\makeatother
\usepackage{geometry}
\usepackage{floatrow}
\usepackage{subfig}
\usepackage{grffile}
\usepackage[labelfont=bf,labelsep=period]{caption}

\usepackage[colorlinks=true,citecolor=blue]{hyperref}
\usepackage[dvips]{epsfig}
\usepackage{color}
\usepackage{bm}
\usepackage{amssymb}
\usepackage{amsmath}
\usepackage{empheq}
\allowdisplaybreaks[4]
\usepackage{amsthm}
\usepackage{etoolbox}
\makeatletter
\patchcmd{\@endtheorem}{\@endpefalse}{}{}{}
\patchcmd{\endproof}{\@endpefalse}{}{}{}
\makeatother
\usepackage{amsfonts}
\usepackage{xspace}
\usepackage{upgreek}
\usepackage{graphicx}
\usepackage{graphicx}
\usepackage{epstopdf}
\usepackage[table]{xcolor}
\usepackage[title]{appendix}
\usepackage{float}
\usepackage{dsfont}
\usepackage{cleveref}

\usepackage{amssymb, mathtools}
\usepackage[table]{xcolor}

\theoremstyle{plain}
\newtheorem{theorem}{Theorem}

\theoremstyle{definition}

\theoremstyle{remark}
\newtheorem{remark}[theorem]{Remark}

\biboptions{sort&compress}
\usepackage{siunitx}
\usepackage{setspace}
\usepackage{esvect}
\usepackage{url}

\DeclareSymbolFont{bbold}{U}{bbold}{m}{n}
\DeclareSymbolFontAlphabet{\mathbbold}{bbold}

\usepackage{chngcntr}
\counterwithout{equation}{section}

\usepackage{accents}

\usepackage{mathtools}

\usepackage{color}

\usepackage[normalem]{ulem} 
\usepackage{xcolor}
\definecolor{forestgreen}{rgb}{0.33,0.61,0.34}

\renewcommand{\theaffn}{\arabic{affn}}

\usepackage{lineno}

\newcommand*{\email}[1]{%
    \footnotesize{*\,Email: #1}\par
    }

\begin{document}

\begin{frontmatter}



\title{Fixation dynamics on hypergraphs}


\author[1]{Ruodan Liu}
\author[1,2]{Naoki Masuda\corref{1}}%
 \address[1]{Department of Mathematics, State University of New York at Buffalo, Buffalo, New York, United States of America}
 \address[2]{Computational and Data-Enabled Sciences and Engineering Program, State University of New York at Buffalo, Buffalo, New York, United States of America\\ \email{naokimas@buffalo.edu}}

\begin{abstract}
%
%
Hypergraphs have been a useful tool for analyzing population dynamics such as opinion formation and the public goods game occurring in overlapping groups of individuals. In the present study, we propose and analyze evolutionary dynamics on hypergraphs, in which each node takes one of the two types of different but constant fitness values. For the corresponding dynamics on conventional networks, under the birth-death process and uniform initial conditions, most networks are known to be amplifiers of natural selection; amplifiers by definition enhance the difference in the strength of the two competing types in terms of the probability that the mutant type fixates in the population. In contrast, we provide strong computational evidence that a majority of hypergraphs are suppressors of selection under the same conditions by combining theoretical and numerical analyses. We also show that this suppressing effect is not explained by one-mode projection, which is a standard method for expressing hypergraph data as a conventional network. Our results suggest that the modeling framework for structured populations in addition to the specific network structure is an important determinant of evolutionary dynamics, paving a way to studying fixation dynamics on higher-order networks including hypergraphs.
\end{abstract}

\begin{keyword}
Evolutionary dynamics, fixation probability, constant selection, hypergraph, amplifier, suppressor
\end{keyword}

\end{frontmatter}

\section*{Author Summary}

Evolutionary dynamics describes spreading and competition of different types of individuals in a population. Prior research has revealed that the population structure, which is typically modeled by networks, is a key factor that affects evolutionary dynamics. Hypergraphs are a generalization of networks and model a set of groups in a population in which a group can involve more than two individuals who simultaneously interact, differently from conventional networks. In the present study, we ask a key question: do hypergraphs yield evolutionary dynamics that are drastically different from those on conventional networks? We have found that the hypergraphs that we have examined are suppressors of natural selection, which discounts the strength of the stronger type towards neutrality. This result is surprising because most conventional networks are amplifiers of natural selection, which magnifies the strength of the stronger type, under the same conditions. Our results suggest that how we model population structure in addition to the specific network structure is an important determinant of evolutionary dynamics.

\section{Introduction}

Populations of individuals, cells, and habitats, on which evolutionary processes take place often have structure that may be described by networks. Evolutionary graph theory enables us to mathematically and computationally investigate population dynamics in which multiple types of different fitness compete on networks under a selection pressure \cite{Nowak2006book, Szabo2007PhyRep, Nowak2010RSocB, Perc2013JRSoc}. A minimal evolutionary dynamics model on graphs, or networks, is to assume that there are two types of different fitness values, $1$ and $r$, which are constant over time (i.e., constant selection), and that each node is occupied by either type, which can change over time. A question then is which type eventually occupies the entire population in the absence of mutation during evolutionary dynamics,  i.e., which type fixates \cite{Lieberman2005nature}. In physics and mathematics literature, the same model is often called the biased voter model, and the fixation is often called the consensus \cite{Durrett1999siam, Antal2006PhysRevLett, Sood2008PhysRevE, Castellano2009RevModPhys}. Crucially, the fixation probability of each type, as well as other properties of the evolutionary dynamics such as the average time to fixation, depends on the network structure in addition to the value of $r$. 
Application of evolutionary graph theory to social dilemma games has also been successful in giving analytical insights into
various conditions under which cooperation is favored over defection \cite{Ohtsuki2006nature, Nowak2006book, Szabo2007PhyRep, Nowak2010RSocB, Allen2017nature}.

Evolutionary graph theory with conventional networks assumes that interaction between nodes occurs pairwise because each edge in a conventional network represents direct connectivity between two nodes.
However, in reality, more than two individuals may simultaneously interact and compete in evolutionary dynamics. For example, an evolutionary dynamics under constant selection can be regarded as a model of opinion dynamics in a population of human or animal individuals, in which they may meet in groups with more than two individuals for opinion formation. As another example, the public goods game is naturally defined for group interaction; each individual in the group decides whether or not to contribute to a collective good, and all individuals will receive a share of an augmented amount of the collective good regardless of whether or not they contributed.
Hypergraphs are a natural extension of conventional networks to the case of group interactions \cite{Battiston2020PhyRep, Lambiotte2019NatPhys, Bianconi2021book}. In a hypergraph, the concept of edge is extended to the hyperedge, which represents a unit of interaction and may contain more than two nodes. Population dynamics on hypergraphs such as evolutionary dynamics of social dilemma games \cite{Santos2008nature, Rodriguez2021natHumBehav} and opinion formation \cite{Sahasrabuddhe2021J.ofPhys, Neuhauser2021PhysRevE} have been investigated
%
%
(see \cite{Battiston2020PhyRep, Majhi2022JRSoc} for reviews including different types of dynamics).

A key question in the study of constant-selection dynamics on networks is whether the given network is amplifier or suppressor of natural selection~\cite{Lieberman2005nature, Nowak2006book}. Suppose that the resident and mutant type have constant fitness values $1$ and $r$, respectively. We anticipate that the mutant type is more likely to fixate than the resident type under the same condition if $r > 1$ and vice versa. In fact, how much the fixation probability of the mutant type increases as one increases $r$ depends on the network structure. Some networks are amplifiers of selection; on these networks, a single mutant has a higher probability of fixation than in the well-mixed population, corresponding to the Moran process, at any $r>1$ and a smaller fixation probability than in the Moran process at any $r<1$. Other networks are suppressors of selection; on these networks, a single mutant has a lower fixation probability than in the Moran process at any $r>1$ and a higher fixation probability than in the Moran process at any $r<1$.
Under the so-called birth-death processes, which we focus on, most networks are known to be amplifiers of selection, at least when the initial mutant is located on a node selected uniformly at random \cite{Hindersin2015PLOS, Cuesta2018PlosOne, Allen2021PlosComputBiol}. Research has discovered various classes of amplifiers of selection \cite{Lieberman2005nature,Giakkoupis2016arxiv,Galanis2017JAcm,Pavlogiannis2017SciRep,Pavlogiannis2018CommBiol,Goldberg2019TheorComputSci} while few for suppressors of selection \cite{Cuesta2017PlosOne}. 

In the present study, we study constant-selection evolutionary dynamics on hypergraphs without mutation. We propose two such models and formulate the fixation probability of mutants using theory based on Markov chains, which extends the same approach for conventional networks to the case of hypergraphs. Then, we examine whether hypergraphs are amplifiers of selection, suppressors of selection, or neither. We show to our numerical efforts that, contrary to the known results for networks, all the hypergraphs that we have investigated are suppressors of selection even under the birth-death process and uniform initial condition. We reach this conclusion by semi-analytical investigations for hypergraphs with high symmetry and numerical simulations on empirical hypergraphs. 

\section{Models}

We introduce two models of evolutionary dynamics on undirected hypergraphs. Let $H$ be an undirected hypergraph with node set $V = \{1, \ldots, N\}$, where $N$ is the number of nodes, and hyperedge set $E$. Each $e \in E$, where $e\subset V$ and $e\neq \emptyset$, is a hyperedge, intuitively representing group interaction among the nodes belonging to $e$. If each $e\in E$ is a set containing exactly two nodes, $H$ is a conventional undirected network.
We assume that $H$ is connected in the sense that there is a hyperedge intersecting both $W$ and $V-W$ for every non-empty proper subset $W$ of $V$.

We define model 1 by the following discrete-time evolutionary dynamics on hypergraph $H$, which extends the birth-death process for conventional networks (precisely speaking, birth-death process with selection on the birth, or the Bd rule \cite{Masuda2009JTB, Shakarian2012Biosys, Pattni2015ProcRSocA}) and the Moran process in well-mixed populations. We assume that there are two types of individuals, referred to as A and B, and that A and B have fitness $r$ and $1$, respectively. We refer to A as the mutant type and B as the resident type. In each time step, we select one node for reproduction with the probability proportional to its fitness. We call this node the parent. Then, the parent node selects one of the hyperedges to which it belongs, denoted by $e$, with the equal probability, regardless of the size of $e$ (i.e., the number of nodes contained in $e$). Finally, the parent converts the type (i.e., A or B) of all other nodes belonging to $e$ into the parent's type. We repeat this process until all nodes in $V$ have the same type. Once this unanimity configuration is reached, which is the fixation, no node will change its type even if one further runs the evolutionary dynamics. We examine the probability that every node eventually has type A, which we call the fixation probability of type A. 

Model 2 is the same as model 1 in that there are two types of individuals with fitness $r$ and $1$, that we select a parent node in each time step with the probability proportional to its fitness, and that the parent selects one of its hyperedges, $e$.
Differently from the model 1, the parent node converts all the other nodes belonging to $e$ into the parent's type if and only if the parent's type is the majority in $e$, i.e., if more than half of the nodes in $e$ including the parent have the same type as the parent's type. Nothing occurs in a time step if the parent's type is not the majority in $e$, including in the case of a tie when the size of $e$ is even. By assumption, a type has to be a majority in an hyperedge for it to spread even when it is stronger than the other type (i.e., type A stronger than type B when $r>1$, and vice versa when $r<1$). Model 2 is an evolutionary dynamics variant of opinion formation models under the majority rule  \cite{Liggett1985book, deOliveira1992JourofStatPhys, Castellano2009RevModPhys}.

It should be noted that the fixation does not necessarily occur in model 2. As an example, consider the hypergraph given by $V = \{ 1, \ldots, 6 \}$ and $E=\{ \{ 1, 2, 3 \}, \{3, 4, 5 \}, \{ 4, 5, 6 \}, \{ 6, 1, 2 \} \}$, and the initial condition in which mutants are only on nodes 1 and 2. Then, under the evolutionary dynamics given by model 2, nodes 4 and 5 will forever be residents, and nodes 3 and 6 will toggle indefinitely.

The assumption that all nodes belonging to a hyperedge $e$ copy the parent's type, made for both models, may seem strong, especially when the size of $e$ is large. In fact, if we allow only one node belonging to $e$ to copy the parent's type in a time step, then it is straightforward to show that model 1 is equivalent to the birth-death process on the conventional network called the one-mode projection (see sections~\ref{sub:one-mode-star} and \ref{sec:numerical} for results for the one-mode projection). Therefore, by assuming that all nodes in $e$ update their type in a time step, we examine effects of an extreme case of allowing simultaneous updating of nodes constrained by the hypergraph structure. 

When the network is the complete graph, which is a conventional undirected network and therefore a hypergraph, model 1 is called the Moran process. For the Moran process, the fixation probability of type A when there are initially $i$ individuals of type A, denoted by $x_i$, is given by \cite{Nowak2006book}
\begin{equation}
x_i=\frac{1-1/r^i}{1-1/r^N}.
\label{eq:Moran-x_i}
\end{equation}

\section{Results for synthetic hypergraphs}

In this section, for three model hypergraphs that are mathematically convenient, we calculate the fixation probability for mutant type A when there are initially $i$ nodes of type A and $N-i$ nodes of type B, for both models 1 and 2. For the cyclic 3-uniform hypergraph, which we introduce in section~\ref{sub:cyclic-model1}, we impose $i=1$ and $i=2$ for models 1 and 2, respectively, for a technical reason to be explained.
The fixation probability depends on the initial condition. We select the $i$ mutant nodes from the $N$ nodes uniformly at random, which is called the uniform initialization \cite{Adlam2015ProcRSocA}.

For any hypergraph of size $N$, at any time step, either type A or B inhabits each node. Therefore, there are $2^N$ states in total. 
The fixation probability for type A depends on each state, not just on $i$. To know the fixation probability for type A for the initial states with $i$ mutants, we need to
solve a linear system of $2^N-2$ unknowns, where each unknown is the fixation probability for an initial state. Note that we safely excluded the two unknowns corresponding to the two initial states in which all nodes unanimously have type A or B. We did so because the fixation probability for type A is trivially equal to 1 and 0 for these two states. To solve a linear system with $2^N-2$ unknowns is daunting except when $N$ is small. 
Therefore, we analyze three types of symmetric hypergraphs in which all or most nodes are structurally equivalent to each other. 
For these hypergraphs, we only need to track the number of nodes with type A among the structurally equivalent nodes. By doing so, we can drastically reduce the dimension of the linear system to be analytically or numerically solved. In this section, we denote the number of nodes of type A in the entire hypergraph by $i \in \{0, \ldots, N\}$. The fixation of type A and B corresponds to $i=N$ and $i=0$, respectively.

\subsection{Model 1\label{sub:model1}}

\subsubsection{Complete 3-uniform hypergraph}

We have mentioned that our model 1 on the complete graph is equivalent to the Moran process. To investigate whether model 1 on counterparts of the complete graph for hypergraphs is equivalent to the Moran process, we consider the complete 3-uniform hypergraph \cite{Verrall1994DisMath}. A complete 3-uniform hypergraph on node set $V$ is defined by hyperedge set $E$, which is the set of all subsets of $V$ containing just three nodes. In other words,
$E = \{ \{v_1, v_2, v_3\} ; v_1, v_2, v_3 \in V, v_1 \neq v_2, v_1 \neq v_3, v_2 \neq v_3 \}$. We show the complete 3-uniform hypergraph on four nodes in Fig~\ref{fig:hypergraphs}A as an example.

\floatsetup[figure]{style=plain,subcapbesideposition=top}
\captionsetup{font={small,rm}} 
\captionsetup{labelfont=bf}
\begin{figure}[H]
  \centering
  \includegraphics[width=1.0\linewidth]{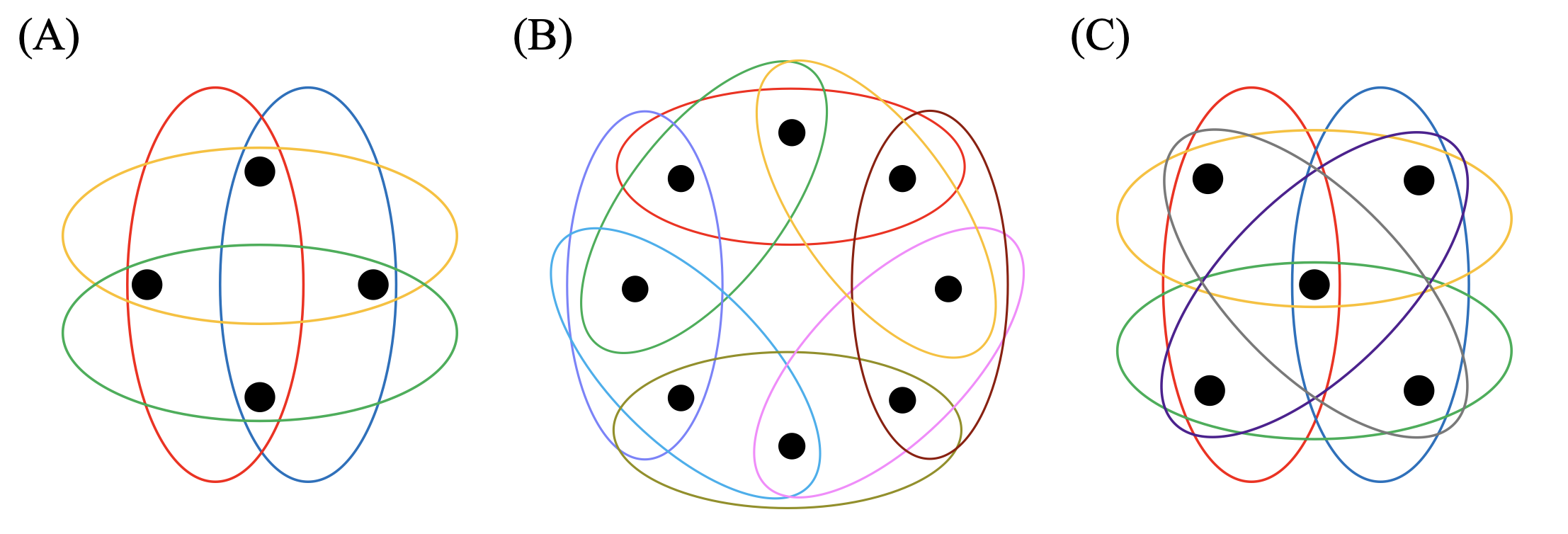}
   \caption{Examples of 3-uniform hypergraphs. (A) Complete 3-uniform hypergraph with $N=4$ nodes. (B) Cyclic 3-uniform hypergraph with $N=8$ nodes. (C) Star 3-uniform hypergraph with $N=5$ nodes. Each colored oval represents a hyperedge.}
   \label{fig:hypergraphs}
\end{figure}

In this section, we refer to $i$, i.e., the number of nodes of type A, as the state of the evolutionary dynamics. Note that knowing the dynamics of $i$ is enough for completely understanding the evolutionary dynamics on the complete 3-uniform hypergraph owing to its symmetry. Under model 1, state $i$ either remains unchanged or moves to $i-2$, $i-1$, $i+1$, or $i+2$ in a single time step of the evolutionary dynamics. This is because every hyperedge of the complete 3-uniform hypergraph has three nodes, and therefore there are at most two nodes that flip their type in a time step.

We denote the $(N+1)\times (N+1)$ transition probability matrix by $P=[p_{i,j}]$, where $p_{i,j}$ is the probability that the state moves from $i$ to $j$ in a time step. At state $i$, the probability that a node of type A and B is selected as parent is equal to $ri/(ri+N-i)$ and $(N-i)/(ri+N-i)$, respectively. If the selected parent node is of type A, a hyperedge containing the parent and two nodes of type B is used for reproduction with probability $\binom{N-i}{2}/\binom{N-1}{2}$, where $\binom{n}{k}$ represents the binomial coefficient ``$n$ choose $k$''. In this case, $i$ increases by $2$.
Alternatively, a hyperedge containing the parent, a different node of type A, and a node of type B is used for reproduction with probability
$\binom{i-1}{1}\binom{N-i}{1}/\binom{N-1}{2}$. In this case, $i$ increases by $1$. Otherwise, a hyperedge containing the parent and two other nodes of type A is used for reproduction. In this case, $i$ does not change.
If the parent is of type B, a hyperedge containing the parent and two nodes of type A is used for reproduction with probability $\binom{i}{2}/\binom{N-1}{2}$. In this case, $i$ decreases by $2$. Alternatively, a hyperedge containing the parent, a node of type A, and a different node of type B is used for reproduction with probability $\binom{i}{1}\binom{N-i-1}{1}/\binom{N-1}{2}$. In this case, $i$ decreases by $1$. Otherwise, a hyperedge containing the parent and two other nodes of type B is selected. In this case, $i$ does not change. Therefore, the transition probabilities are given by
\begin{linenomath}
\begin{align}
p_{0,0}&=1,\\
p_{N,N}&=1,\\
p_{i,i-2}&=\frac{N-i}{ri+N-i}\cdot\frac{\binom{i}{2}}{\binom{N-1}{2}}=\frac{N-i}{ri+N-i}\cdot\frac{i(i-1)}{(N-1)(N-2)}, \quad i\in \{2, \ldots, N\},\\
p_{i,i-1}&=\frac{N-i}{ri+N-i}\cdot\frac{\binom{i}{1}\binom{N-i-1}{1}}{\binom{N-1}{2}}=\frac{N-i}{ri+N-i}\cdot\frac{2i(N-i-1)}{(N-1)(N-2)}, \quad i\in \{1, \ldots, N\},\\
p_{i,i+1}&=\frac{ri}{ri+N-i}\cdot\frac{\binom{i-1}{1}\binom{N-i}{1}}{\binom{N-1}{2}}=\frac{ri}{ri+N-i}\cdot\frac{2(i-1)(N-i)}{(N-1)(N-2)}, \quad i\in \{2, \ldots, N-1\},\\
p_{i,i+2}&=\frac{ri}{ri+N-i}\cdot\frac{\binom{N-i}{2}}{\binom{N-1}{2}}=\frac{ri}{ri+N-i}\cdot\frac{(N-i)(N-i-1)}{(N-1)(N-2)}, \quad i\in \{1, \ldots, N-2\},\\
p_{1,1}&=1-p_{1,0}-p_{1,2}-p_{1,3},\\
p_{i,i}&=1-p_{i,i-2}-p_{i,i-1}-p_{i,i+1}-p_{i,i+2}, \quad i\in \{2, \ldots, N-2\},\\
p_{N-1,N-1}&=1-p_{N-1,N-3}-p_{N-1,N-2}-p_{N-1,N}.
\end{align}
\end{linenomath}
All the other entries of $P$ are equal to zero. Therefore, $P$ is a pentadiagonal matrix. States $i=0$ and $i=N$ are absorbing states.

Denote by $x_i$ the probability of ending up in state $N$, i.e., fixation probability of the mutant type A, when the initial state is $i$. We obtain
\begin{linenomath}
\begin{align}
x_0&=0,
\label{complete-x0}\\
x_1&=p_{1,0}x_{0}+p_{1,1}x_{1}+p_{1,2}x_{2}+p_{1,3}x_{3},\label{complete-x1}\\
x_i&=p_{i,i-2}x_{i-2}+p_{i,i-1}x_{i-1}+p_{i,i}x_{i}+p_{i,i+1}x_{i+1}+p_{i,i+2}x_{i+2}, \quad i\in \{2, \ldots, N-2\},
\label{eq:recursive-complete}\\
x_{N-1}&=p_{N-1,N-3}x_{N-3}+p_{N-1,N-2}x_{N-2}+p_{N-1,N-1}x_{N-1}+p_{N-1,N}x_{N},\label{complete-xn-1}\\
x_N&=1.
\label{complete-xn}
\end{align}
\end{linenomath}
In vector notation, we can concisely write Eqs.~\labelcref{complete-x0,complete-x1,eq:recursive-complete,complete-xn-1,complete-xn} as
\begin{equation}\label{veceq}
\bm{x}=P\bm{x},
\end{equation}
where $\bm{x}=(x_0, x_1, x_2, \ldots, x_N)^{\top}$, and $^{\top}$ represents the transposition.
Equation~\eqref{veceq} is equivalent to
\begin{equation}\label{homo}
(P-I)\bm{x}=\bm{0},
\end{equation}
where $I$ is the identity matrix. Using Eqs.~\eqref{complete-x0}, \eqref{complete-xn}, and \eqref{homo}, we obtain
\begin{equation}\label{eq:Mx=b}
M\bm{x}=\bm{b},
\end{equation}
where $\bm{b}=(0, 0, \ldots, 0, 1)^{\top}$, and
\begin{equation}
M=
\begin{pmatrix}
1 & 0 & 0 & 0 & 0 & \cdots & 0 & 0 & 0 & 0\\
p_{1,0} & p_{1,1}-1 & p_{1,2} & p_{1,3} & 0 & \cdots & 0 & 0 & 0 & 0\\
p_{2,0} & p_{2,1} & p_{2,2}-1 & p_{2,3} & p_{2,4} & \cdots & 0 & 0 & 0 & 0\\
\vdots & \vdots & \vdots & \vdots & \vdots & \cdots & \vdots & \vdots & \vdots & \vdots\\
0 & 0 & 0 & 0 & 0 & \cdots & p_{N-1,N-3} & p_{N-1,N-2} & p_{N-1,N-1}-1 & p_{N-1,N}\\
0 & 0 & 0 & 0 & 0 & \cdots & 0 & 0 & 0 & 1
\end{pmatrix}.
\end{equation}
Like $P$, the $(N+1)\times (N+1)$ matrix $M$ is a pentadiagonal matrix. 

PTRANS-I and PTRANS-II are numerical algorithms for efficiently solving pentadiagonal linear systems \cite{Askar2015}. 
Note that we are calculating $\bm x$ although we only need $x_1$. These two algorithms run in O($\log n$) time, where $n$ is the number of unknowns, and, for $n=4000$, they are about ten times faster than the SciPy algorithm for banded matrices, scipy.linalg.solve\underline{~~}banded \cite{scipy.linalg.solve.banded}. We use the PTRANS-II built in a Python package pentapy \cite{Muller2019pentapy} to calculate the fixation probability as a function of $r$ unless $N$ is small. For a given value of $r$, the calculation requires $O(\log N)$ time, which is much faster than solving a full linear system of $2^N-2$ unknowns.

For small complete 3-uniform hypergraphs, one can analytically calculate $x_1$ by directly solving Eq.~\eqref{eq:Mx=b}.
We obtain
\begin{equation}
x_1 = \frac{r^2}{r^2+2r+1}
\label{eq:complete_N=4}
\end{equation}
for $N=4$ and
\begin{equation}
x_1 = \frac{r^2(8r^2+12r+1)}{8r^4+28r^3+33r^2+28r+8}
\label{eq:complete_N=5}
\end{equation}
for $N=5$. We show $x_i$ for $i \in \{2, \ldots, N-1\}$, $N \in \{4, 5\}$ in Text A in S1 File. We compare Eqs.~\eqref{eq:complete_N=4} and \eqref{eq:complete_N=5} with $x_1$ for the Moran process (i.e., Eq.~\eqref{eq:Moran-x_i} with $i=1$) with $N=4$ and $N=5$ as a function of $r$ in Fig~\ref{fig:m1_comparison}A and \ref{fig:m1_comparison}B, respectively. Although the results are shown by blue lines, they are hidden behind the orange lines. The figure indicates that these complete 3-uniform hypergraphs are suppressors of selection because their $x_1$ is smaller than that for the Moran process for $r>1$ and larger for $r<1$. We also obtained $x_1$ by numerically solving Eq.~\eqref{eq:Mx=b} for $N=20$ and $N=200$. The results shown in Fig~\ref{fig:m1_comparison}C and \ref{fig:m1_comparison}D for $N=20$ and $N=200$, respectively, indicate that these larger complete 3-uniform hypergraphs are also suppressors of selection. Therefore, we conclude that the complete 3-uniform hypergraphs are suppressors of selection under the evolutionary dynamics described by model 1.

\floatsetup[figure]{style=plain,subcapbesideposition=top}
\captionsetup{font={small,rm}} 
\captionsetup{labelfont=bf}
\begin{figure}[t]
  \centering
  \includegraphics[width=0.9\linewidth]{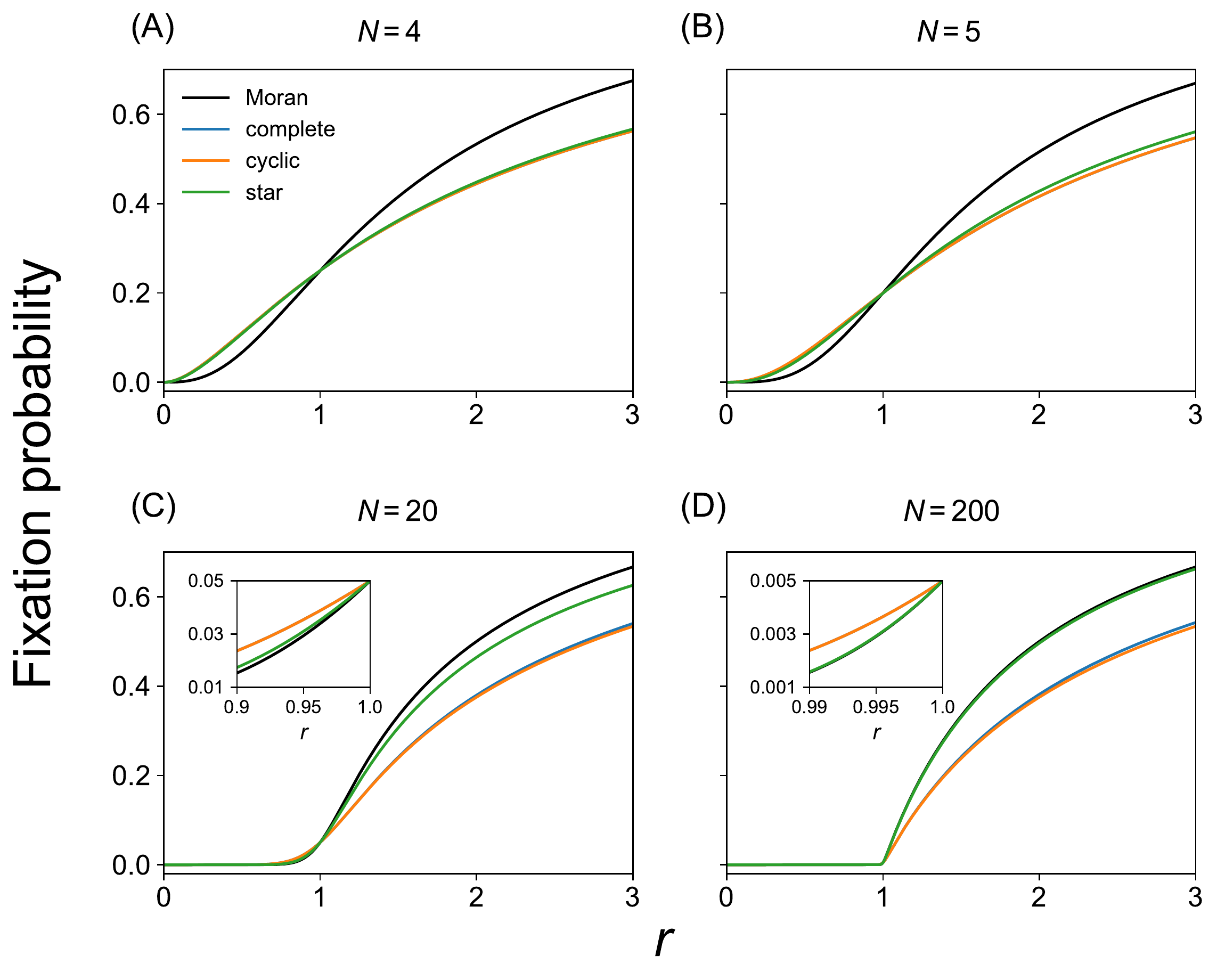}
   \caption{Fixation probability for different hypergraph models under model 1. We compare the Moran process, which is the baseline, complete 3-uniform hypergraphs, cyclic 3-uniform hypergraphs, and star 3-uniform hypergraphs. (A) $N=4$. (B) $N=5$. (C) $N=20$. (D) $N=200$. The insets in (C) and (D) magnify the results for $r$ values smaller than and close to $r=1$. In these insets and main panel (B), the results for the complete 3-uniform hypergraph are close to those for the cyclic hypergraph such that the blue lines are almost hidden behind the orange lines. In (A), the two results are exactly the same such that the blue line is completely hidden behind the orange line.
}
   \label{fig:m1_comparison}
\end{figure}

\subsubsection{Cyclic 3-uniform hypergraph\label{sub:cyclic-model1}}

The fixation probability as a function of $r$ for undirected cycle graphs is the same as that for the Moran process because the cycle graphs are so-called isothermal graphs under the birth-death process \cite{Nowak2006book, Lieberman2005nature, Pattni2015ProcRSocA}.
To examine whether the same equivalence result holds true for hypergraphs, we consider an extension of the cycle graph to the hypergraph, which we call the cyclic 3-uniform hypergraph. A cyclic 3-uniform hypergraph consists of a node set $V= \{1, \ldots, N\}$ and a hyperedge set $E = \{\{1,2,3\}, \{2,3,4\},\ldots, \{N-2, N-1, N\}, \{N-1, N,1\}, \{N, 1, 2\}\}$, i.e., any three consecutively numbered nodes with a periodic boundary condition form a hyperedge. The cyclic 3-uniform hypergraph with $N=3$ nodes is the complete 3-uniform hypergraph. Therefore, we consider cyclic 3-uniform hypergrapys with $N\ge 4$. We show the cyclic 3-uniform hypergraph on 8 nodes in Fig~\ref{fig:hypergraphs}B.

We assume that there is initially one individual of type A. Then, under model 1, all the nodes of type A are consecutively numbered at any time, without being interrupted by nodes of type B. Therefore, similarly to the analysis of evolutionary dynamics on cycles \cite{Ohtsuki2006ProcRSocB, Broom2008ProcRSocA}, it suffices to track the number of nodes of type A, which we again denote by $i$, to understand the evolutionary dynamics on this hypergraph. It should be noted that this technique can only be used when $i=1$ or when the $i$ $(\ge 2)$ nodes of type A are consecutively located on the cycle in the initial state.

\floatsetup[figure]{style=plain,subcapbesideposition=top}
\captionsetup{font={small,rm}} 
\captionsetup{labelfont=bf}
\begin{figure}[H]
  \centering
  \includegraphics[width=1.0\linewidth]{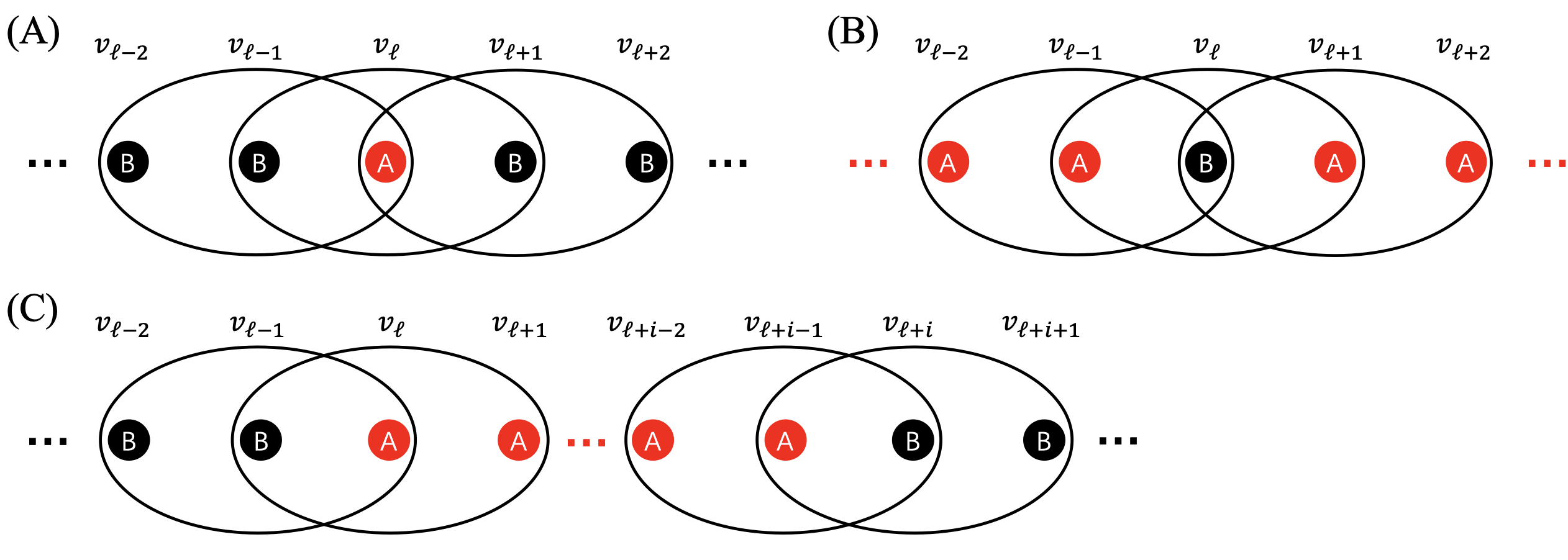}
   \caption{State transitions in the cyclic 3-uniform hypergraph. (A) A state with just one node of type A (i.e., $i=1$). (B) A state with just one node of type B (i.e., $i=N-1$). (C) A state with more than one nodes of each type (i.e., $2\le i\le N-2$).}
   \label{fig:cyclic}
\end{figure}

When $i=1$, there are three types of events that can occur next. Without loss of generality, we assume that the $\ell$th node is the only node of type A (see Fig~\ref{fig:cyclic}A for a schematic). The state moves from $i=1$ to $i=0$ in one time step such that B fixates in the following two cases.
In the first case, either node $v_{\ell-2}$ or $v_{\ell+2}$, which is of type B, is selected as parent. If $N\ge 5$, these two nodes are distinct. Therefore, this event occurs with probability $2/(r+N-1)$.
Then, the hyperedge that contains the parent and $v_{\ell}$ is used for reproduction, which occurs with probability $1/3$. 
For example, if $v_{\ell-2}$ is selected as parent and hyperedge $\{\ell-2, \ell-1, \ell\}$ is used for reproduction, then the state moves from $1$ to $0$.
In the second case, either $v_{\ell-1}$ or $v_{\ell+1}$ is selected as parent, which occurs with probability $2/(r+N-1)$. Then, one of the two hyperedges that contain the parent and $v_{\ell}$ is used for reproduction, which occurs with probability $2/3$. For example, if $v_{\ell-1}$ is selected as parent and hyperedge $\{ \ell-2, \ell-1, \ell \}$ or $\{ \ell-1, \ell, \ell+1 \}$ is used for reproduction, then the state moves from $1$ to $0$.
By summing up these probabilities, we obtain
\begin{equation}
p_{1,0}=\frac{2}{r+N-1}.
\label{eq:model1-cycle-i=1-decrease1}
\end{equation}
In fact, Eq.~\eqref{eq:model1-cycle-i=1-decrease1} also holds true for $N=4$ although $v_{\ell-2}$ and $v_{\ell+2}$ are identical nodes when $N=4$; see~Text B in S1 File for the proof.
Alternatively, the state moves from $i=1$ to $i=3$ whenever $v_{\ell}$ is selected as parent, which occurs with probability $r/(r+N-1)$. Therefore, we obtain
\begin{equation}
p_{1,3}=\frac{r}{r+N-1}.
\end{equation}
If any other event occurs, then $i=1$ remains unchanged. Therefore, we obtain
\begin{linenomath}
\begin{align}
p_{1,1} &= 1-p_{1,0}-p_{1,3}=\frac{N-3}{r+N-1},\\
p_{1,j} &= 0 \text{ if } j \neq 0 , 1, 3.
\end{align}
\end{linenomath}

When $i=N-1$, without loss of generality, we assume that the $\ell$th node is the only node of type B (see Fig~\ref{fig:cyclic}B for a schematic). As shown in Text C in S1 File, we obtain
\begin{linenomath}
\begin{align}
p_{N-1,N-3}&=\frac{1}{r(N-1)+1}, 
\label{eq:model1-cycle-i=n-1-decrease2}\\
p_{N-1,N}&=\frac{2r}{r(N-1)+1},
\label{eq:model1-cycle-i=n-1-increase1}\\
p_{N-1,N-1}&=1-p_{N-1,N-3}-p_{N-1,N}=\frac{r(N-3)}{r(N-1)+1},
\label{eq:model1-cycle-i=n-1-unchanged}\\
p_{N-1, j}&=0 \text{ if } j \neq N-3 , N-1, N.
\label{eq:model1-cycle-i=n-1-impossible}
\end{align}
\end{linenomath}

When $i\in \{2, \ldots, N-2\}$, without loss of generality, we assume that the $\ell$th to the $(\ell+i-1)$th nodes are of type A and that all other nodes are of type B (see Fig~\ref{fig:cyclic}C for a schematic).
As shown in Text C in S1 File, we obtain
\begin{linenomath}
\begin{align}
p_{i,i-2}&=\frac{2}{3(ri+N-i)}, 
\label{eq:model1-cycle-i=i-decrease2}\\
p_{i,i-1}&=\frac{4}{3(ri+N-i)},
\label{eq:model1-cyclic-decrease1}\\
p_{i,i+1}&=\frac{4r}{3(ri+N-i)},
\label{eq:model1-cyclic-increase1}\\
p_{i,i+2}&=\frac{2r}{3(ri+N-i)},
\label{eq:model1-cycle-i=i-increase2}
\end{align}
\end{linenomath}
and
\begin{equation}
p_{i,i}=1-p_{i,i-2}-p_{i,i-1}-p_{i,i+1}-p_{i,i+2}.
\label{eq:model1-cycle-i=i-unchanged}
\end{equation}
All the other entries of transition probability matrix $P$ are equal to zero. 

By solving Eq.~\eqref{eq:Mx=b}, we obtain
\begin{equation}\label{eq:cyclic_N=4}
x_1 = \frac{r^2}{r^2+2r+1}
\end{equation}
for $N=4$ and
\begin{equation}\label{eq:cyclic_N=5}
x_1 = \frac{r^2(6r^2+8r+1)}{6r^4+20r^3+23r^2+20r+6}
\end{equation}
for $N=5$.
We show $x_i$ for $i \in \{2, \ldots, N-1\}$, $N \in \{4, 5\}$ in Text A in S1 File. We compare Eqs.~\eqref{eq:cyclic_N=4} and \eqref{eq:cyclic_N=5} with $x_1$ for the Moran process with $N=4$ and $N=5$ in Fig~\ref{fig:m1_comparison}A and \ref{fig:m1_comparison}B, respectively. We find that these cyclic 3-uniform hypergraphs are suppressors of selection. The result for $N=4$ given by Eq.~\eqref{eq:cyclic_N=4} coincides with that for the complete 3-uniform hypergraph given by Eq.~\eqref{eq:complete_N=4}. For larger $N$, we use the same numerical method for solving Eq.~\eqref{eq:Mx=b} as that for the complete 3-uniform hypergraph because $P$ is a pentadiagonal matrix. We show the thus calculated $x_1$ for $N=20$ and $N=200$ in Fig~\ref{fig:m1_comparison}C and \ref{fig:m1_comparison}D, respectively. These figures confirm that these larger cyclic 3-uniform hypergraphs are also suppressors of selection. Therefore, we conclude that the cyclic 3-uniform hypergraph is a suppressor of selection.

\subsubsection{Star 3-uniform hypergraph\label{sub:model1-star}}

The conventional star graphs are strong amplifiers of selection \cite{Lieberman2005nature, Pavlogiannis2018CommBiol}. To examine whether counterparts of the star graph for hypergraphs are also amplifiers of selection, we define the star 3-uniform hypergraph as follows and examine the fixation probability of the mutant on it.
A star 3-uniform hypergraph consists of a node set $V$ with a single hub node and $N-1$ leaf nodes, and hyperedges each of which is
of size three and consists of the hub node and a pair of $N-1$ leaf nodes. We use all the $\binom{N-1}{2}$ pairs of leaf nodes to form hyperedges such that the hub node belongs to $\binom{N-1}{2}$ hyperedges. Each leaf node belongs to $N-1$ hyperedges and are structurally equivalent to each other. The star 3-uniform hypergraph with $N=3$ is the complete 3-uniform hypergraph. Therefore, we consider star 3-uniform hypergraphs with $N\ge 4$. We show the 3-uniform star hypergraph on 5 nodes in Fig~\ref{fig:hypergraphs}C as an example.

Owing to the structural equivalence among the $N-1$ leaf nodes, the state of the evolutionary dynamics on the star 3-uniform hypergraph is completely determined by the type on the hub (i.e., either A or B) and the number of leaf nodes of type A, which ranges between $0$ and $N-1$. Therefore, we denote the state by
$(i_1, i_2)$, where $i_1=1$ or 0 if the hub node is of type A or B, respectively, and $i_2 \in \{0, 1, \ldots, N-1 \}$ represents the number of the leaf nodes of type A. 
The total number of nodes of type A is given by $i=i_1+i_2$. The fixation of type A and B corresponds to $(i_1, i_2) = (1, N-1)$ and $(0, 0)$, respectively.

We denote the probability that type A fixates starting with state $(i_1, i_2)$ by $\tilde{x}_{(i_1,i_2)}$. 
As shown in Text G in S1 File, we obtain the vector of fixation probability,
$\tilde{\bm{x}} \equiv (\tilde{x}_{(0, 0)}, \tilde{x}_{(0, 1)}, \ldots, \tilde{x}_{(0, N-1)}, \tilde{x}_{(1, 0)}, \tilde{x}_{(1, 1)}, \ldots, \tilde{x}_{(1, N-1)})^{\top}$, 
as follows:
\begin{equation}\label{veceqstar}
\tilde{\bm{x}}=P\tilde{\bm{x}},
\end{equation}
where the $2N\times 2N$ stochastic matrix $P$ is given by
\begin{equation}
P=
\left(\begin{array}{ c | c }
    C & D \\
    \hline
    E & F
\end{array}\right),
\label{eq:P-block}
\end{equation}
and $C, D, E,$ and $F$ are $N\times N$ matrices given in Text G in S1 File.
The transition matrix $P$ is a sparse matrix. We use the scipy implementation of the DGESV routine of LAPACK, scipy.linalg.solve, to numerically solve Eq.~\eqref{veceqstar} to obtain $\tilde{x}_{(0, i)}$ and $\tilde{x}_{(1, i-1)}$. 

We remind that $x_i$ is the probability that type A fixates when there are initially $i$ nodes of type A that are selected uniformly at random. There are $\binom{N-1}{i-1}$ states in which $i-1$ nodes have type A and the hub has type A. There are $\binom{N-1}{i}$ states in which $i$ nodes have type A and the hub has type B. Therefore, we obtain
\begin{linenomath}
\begin{align}\label{average_over}
x_i =& \frac{\binom{N-1}{i-1}} {\binom{N-1}{i-1} + \binom{N-1}{i}} \tilde{x}_{(1, i-1)}
+ \frac{\binom{N-1}{i}} {\binom{N-1}{i-1} + \binom{N-1}{i}} \tilde{x}_{(0, i)}\notag\\
=& \frac{i}{N}\tilde{x}_{(1, i-1)}+\frac{N-i}{N}\tilde{x}_{(0, i)}.
\end{align}
\end{linenomath}

We analytically solve Eqs.~\eqref{veceqstar} and \eqref{average_over} to obtain
\begin{equation}\label{eq:star_N=4}
x_1 = \frac{r^2(3r+5)}{4(3r^3+14r^2+18r+9)}+\frac{9r^2}{4(3r^2+5r+3)}
\end{equation}
for $N=4$ and
\begin{linenomath}
\begin{align}
x_1 =& \frac{r^2(72r^3+202r^2+145r+6)}{5(72r^5+490r^4+1025r^3+1070r^2+880r+288)}\nonumber\\
&+\frac{4r^2(72r^2+94r+4)}{5(72r^4+202r^3+217r^2+202r+72)} \label{eq:star_N=5}
\end{align}
\end{linenomath}
for $N=5$. We show $x_i$ for $i \in \{2, \ldots, N-1\}$, $N \in \{4, 5\}$ in Text A in S1 File. We compare Eqs.~\eqref{eq:star_N=4} and \eqref{eq:star_N=5} with $x_1$ for the Moran process with $N=4$ and $N=5$ in Fig~\ref{fig:m1_comparison}A and \ref{fig:m1_comparison}B, respectively. We find that these star 3-uniform hypergraphs are suppressors of selection. We then computed $x_1$ by numerically solving Eq.~\eqref{veceqstar} for $N=20$ and $N=200$. Fig~\ref{fig:m1_comparison}C and \ref{fig:m1_comparison}D compare the obtained $x_1$ values with those for the Moran process with $N=20$ and $N=200$, respectively. These figures indicate that these larger star 3-uniform hypergraphs are also suppressors of selection although the degree of suppression is smaller for larger star 3-uniform hypergraphs; Fig~\ref{fig:m1_comparison}D shows that the difference between the star 3-uniform hypergraph and the Moran process in terms of $x_1$ is tiny when $N=200$. Fig~\ref{fig:m1_comparison} suggests that larger star 3-uniform hypergraphs are weaker suppressors of selection.
We also verified that the star 3-uniform hypergraph with $N=1500$ nodes is also a suppressor of selection although the suppressing effect is even weaker (see Text H in S1 File).
 This result is consistent with the analytically known result that larger conventional star graphs are stronger amplifiers of selection \cite{Lieberman2005nature, Pavlogiannis2018CommBiol}. Therefore, we conclude based on our numerical results that the star 3-uniform hypergraph is likely to be a suppressor of selection for any $N$.

\subsubsection{Comparison with conventional networks obtained by one-mode projection\label{sub:one-mode-star}}

A lossy representation of a hypergraph as a conventional network is the one-mode projection, in which two nodes are adjacent by an edge if and only if they belong to at least one common hyperedge. A weighted network version of the one-mode projection defines the edge weight by the number of hyperedges that the two nodes share. The one-mode projection is a common approach for analyzing dynamics on hypergraphs including evolutionary dynamics
\cite{Perc2013JRSoc}.
The one-mode projections of the complete and cyclic 3-uniform hypergraphs are regular networks (i.e., networks in which all the nodes have the same degree), in the case of both unweighted and weighted variants of the one-mode projection. Then, the isothermal theorem \cite{Lieberman2005nature} guarantees that the birth-death process on the one-mode projections of the complete and cyclic 3-uniform hypergraphs is equivalent to the Moran process. Therefore, our result that the complete and cyclic 3-uniform hypergraphs are suppressors of selection is not an artifact of one-mode projection. 

The same argument does not apply for the star 3-uniform hypergraph because the one-mode projection of the star 3-uniform hypergraph is not a regular network. Furthermore, the star 3-uniform hypergraph are only analogously similar to the conventional star graph.
In fact, leaf nodes are adjacent to each other in the star 3-uniform hypergraph, whereas they are not in the conventional star graph. 
Then, the direct connection between leaf nodes might be a reason why the star 3-uniform hypergraph is a suppressor of selection.
To exclude this possibility, using similar analytical techniques, we investigated the fixation probability for the weighted one-mode projection of the star 3-uniform hypergraph, which is a weighted complete graph. Note that the unweighted one-mode projection of the star 3-uniform hypergraph is the unweighted complete graph, which is trivially equivalent to the Moran process. As we show in Text I in S1 File,
we have found that the obtained weighted complete graph is an amplifier of selection although the amplifying effect is weak. Therefore, our result that the star 3-uniform hypergraph is a suppressor of selection is not expected from the one-mode projection.

\subsection{Model 2\label{sub:model2}}

In this section, we analyze the fixation probability for the birth-death process governed by model 2 on the complete, cyclic, and star 3-uniform hypergraphs.

\subsubsection{Complete 3-uniform hypergraph\label{sub:complete-model2}}

On the complete 3-uniform hypergraph, the state $i$ either remains unchanged or moves to $i-1$ or $i+1$ in a single time step because all the hyperedges are composed of three nodes such that there are at most one node that changes the state under the majority rule given by model 2. At state $i$, the probability that a node of type A and B is selected as parent is equal to $ri/(ri+N-i)$ and $(N-i)/(ri+N-i)$, respectively. If the parent is of type A, then the following three types of events are possible. First, if a hyperedge containing the parent and two nodes of type B is used for reproduction with probability $\binom{N-i}{2}/\binom{N-1}{2}$, then the state does not change.
Second, if a hyperedge containing the parent, a different node of type A, and a node of type B is used for reproduction with probability
$\binom{i-1}{1}\binom{N-i}{1}/\binom{N-1}{2}$, then the state moves from $i$ to $i+1$.
Third, if a hyperedge containing the parent and two other nodes of type A is used for reproduction with the remaining probability, then the state does not change.
If the parent is of type B, then the following three types of events are possible. First, if a hyperedge containing the parent and two nodes of type A is used for reproduction with probability $\binom{i}{2}/\binom{N-1}{2}$, then the state does not change. Second, if a hyperedge containing the parent, a node of type A, and a different node of type B is used for reproduction with probability $\binom{i}{1}\binom{N-i-1}{1}/\binom{N-1}{2}$, then the state moves from $i$ to $i-1$. Third, if a hyperedge containing the parent and two other nodes of type B is used for reproduction with the remaining probability, then the state does not change. Therefore, the transition probabilities are given by
\begin{linenomath}
\begin{align}
p_{0,0}&=1,\\
p_{N,N}&=1,\\
p_{i,i-1}&=\frac{N-i}{ri+N-i}\cdot\frac{\binom{i}{1}\binom{N-i-1}{1}}{\binom{N-1}{2}}=\frac{N-i}{ri+N-i}\cdot\frac{2i(N-i-1)}{(N-1)(N-2)}, \quad i\in \{1, \ldots, N\},\\
p_{i,i+1}&=\frac{ri}{ri+N-i}\cdot\frac{\binom{i-1}{1}\binom{N-i}{1}}{\binom{N-1}{2}}=\frac{ri}{ri+N-i}\cdot\frac{2(i-1)(N-i)}{(N-1)(N-2)}, \quad i\in \{0, \ldots, N-1\},\\
p_{i,i}&=1-p_{i,i-1}-p_{i,i+1}, \quad i\in \{1, \ldots, N-1\}.
\end{align}
\end{linenomath}
All the other entries of $P$ are equal to zero. Therefore, $P$ is a tridiagonal matrix. We note that the states $i=0$ and $i=N$ are absorbing states.

The fixation probability of type A starting from state $i$, i.e., $x_i$, satisfies
\begin{linenomath}
\begin{align}
&x_0=x_1=0,
\label{x0}\\
&x_i=p_{i,i-1}x_{i-1}+p_{i,i}x_{i}+p_{i,i+1}x_{i+1}, \quad i\in \{2, \ldots, N-2\},
\label{eq:recursive}\\
&x_{N-1}=x_N=1.
\label{xn}
\end{align}
\end{linenomath}
Similar to the analysis of the fixation probability for the Moran process \cite{Nowak2006book}, we set
\begin{equation}
y_i \equiv x_i-x_{i-1}, \quad i\in \{1, \ldots, N\}.
\end{equation}
Note that $\sum_{i=1}^N y_i = x_N-x_0=1$. Let $\gamma_i=p_{i,i-1}/p_{i,i+1}$. Equation~\eqref{eq:recursive} leads to $y_{i+1}=y_i \gamma_i$ with $i\in \{2, \ldots, N-2\}$. Therefore,  we obtain $y_1=0$, $y_2=x_2$, $y_3=x_2 \gamma_2$, $y_4= x_2 \gamma_2\gamma_3$, $\ldots$, $y_{N-1} = x_2 \prod_{k=2}^{N-2} \gamma_k$.
By summing all these expressions and using $y_N = 0$, which Eq.~\eqref{xn} implies, we obtain
\begin{linenomath}
\begin{align}
1=\sum_{i=1}^N y_i&=x_2\left(1+\gamma_2+\gamma_2\gamma_3+\cdots+\prod_{k=2}^{N-2} \gamma_k\right) \notag\\
&=x_2\left[1+\frac{N-3}{r}+\frac{N-4}{2r}\frac{N-3}{r}+\cdots+\frac{1}{r(N-3)}\frac{2}{r(N-4)}\cdots\frac{N-4}{2r}\frac{N-3}{r}\right] \notag\\
%
%
&=x_2\sum_{i=0}^{N-3}r^{-i}\binom{N-3}{i} \notag\\
&=x_2\left(1+\frac{1}{r}\right)^{N-3}.
\end{align}
\end{linenomath}
Therefore, we obtain
\begin{equation}
x_2=\left(1+\frac{1}{r}\right)^{3-N}
\label{eq:model2-complete-x2}
\end{equation}
and
\begin{linenomath}
\begin{align}
x_i =& x_2\left(1+\sum_{j=2}^{i-1}\prod_{k=2}^j \gamma_k\right)\notag\\
& = \left(1+\frac{1}{r}\right)^{3-N}\left(1+\sum_{j=2}^{i-1}\prod_{k=2}^j \gamma_k\right), \quad i\in \{3, 4, \ldots, N-2\}.
\end{align}
\end{linenomath}
This expression does not simplify further.

For model 2, it always holds true that $x_1 = 0$ because the evolutionary dynamics is driven by a majority rule. Therefore, we compare $x_2$, instead of $x_1$, as a function of $r$ with $x_2$ for the Moran process, to examine whether a given hypergraph is an amplifier of selection, suppressor of selection, equivalent to the Moran process, or neither. We compare $x_2$ calculated from Eq.~\eqref{eq:model2-complete-x2} with that for the Moran process at four values of $N$ in Fig~\ref{fig:m2_comparison}. The figure indicates that the complete 3-uniform hypergraph with 4 nodes is a suppressor of selection. However, the complete 3-uniform hypergraph with $N > 4$ is neither amplifier nor suppressor of selection. This is because Eq.~\eqref{eq:model2-complete-x2} implies that $x_2=2^{3-N}$ when $r=1$, which is different from the result for the Moran process, i.e., $x_2=2/N$, when $N > 4$.

\floatsetup[figure]{style=plain,subcapbesideposition=top}
\captionsetup{font={small,rm}} 
\captionsetup{labelfont=bf}
\begin{figure}[H]
  \centering
  \includegraphics[width=0.9\linewidth]{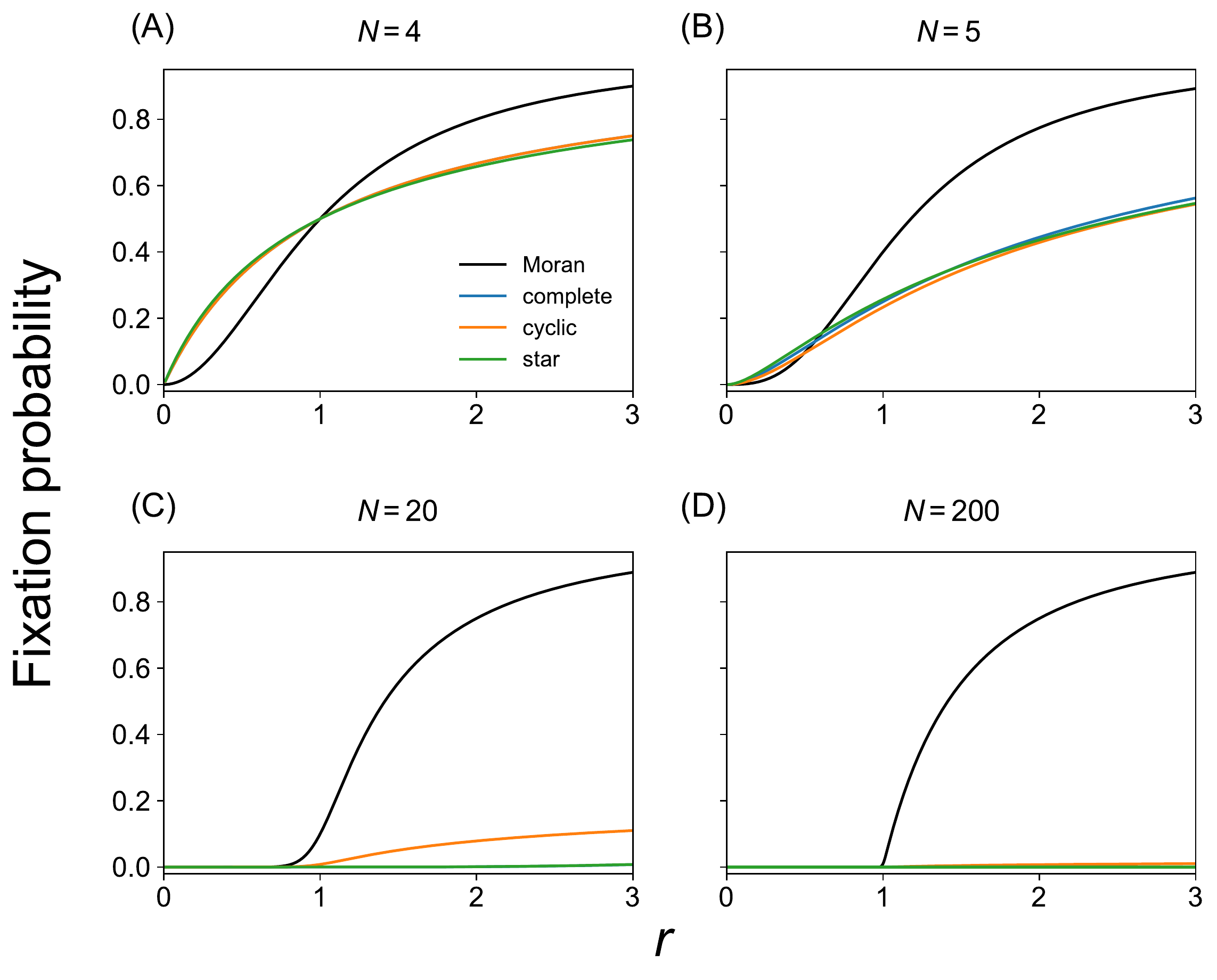}
   \caption{Fixation probability for different hypergraph models under model 2. We compare the Moran process, which is the baseline, complete 3-uniform hypergraphs, cyclic 3-uniform hypergraphs, and star 3-uniform hypergraphs. (A) $N=4$. (B) $N=5$. (C) $N=20$. (D) $N=200$. In (A), the result for the complete 3-uniform hypergraph is exactly the same as that for the cyclic 3-uniform hypergraph such that the blue line is completely hidden behind the orange line. In (C) and (D), the results for the complete 3-uniform hypergraph (shown by the blue lines) are not identical but close to those for the star hypergraph (shown by the green lines) such that the former are hidden behind the latter.}
   \label{fig:m2_comparison}
\end{figure}

\subsubsection{Cyclic 3-uniform hypergraph\label{sub:cyclic-model2}}

Consider the evolutionary dynamics under model 2 on the cyclic 3-uniform hypergraph. Assume that there are initially two nodes of type A, which are distributed uniformly at random, and $N-2$ nodes of type B. We derived in~Text J in S1 File the fixation probability for type A, $x_2$, using techniques similar to those used for the complete 3-uniform hypergraph (section~\ref{sub:complete-model2}). We obtain
\begin{equation}
x_2 = \begin{cases}
\frac{2}{N-1}\left\{\frac{1-r^{-1}}{1-r^{-(N-2)}}+\frac{r-r^{-1}}{(r+4)\left[1-r^{-(N-2)}\right]}\right\}
& (N\geq 5),\\[1mm]
\frac{r}{1+r} & (N=4).
\end{cases}
\label{eq:model2-cyclic-x2}
\end{equation}
We show $x_3$ for $N = 5$ in Text A in S1 File; note that $x_2$ for $N \in \{4, 5\}$ and $x_3$ for $N = 5$ are the only nontrivial $x_i$ values for $N \in \{4, 5 \}$ because $x_1=0$ and $x_{N-1}=1$ for model 2.
 
We compare $x_2$ given by Eqs.~\eqref{eq:model2-cyclic-x2} with $x_2$ for the Moran process at four values of $N$ in Fig~\ref{fig:m2_comparison}. The figure indicates that the cyclic 3-uniform hypergraph with $N=4$ is a suppressor of selection. However, the cyclic 3-uniform hypergraph with $N>4$ is neither amplifier nor suppressor of selection under model 2. This is because $x_2=14/\left[5(N-1)(N-2)\right]$ when $r=1$ for cyclic 3-uniform hypergraphs. This $x_2$ value is different from the value for the Moran process, i.e., $x_2=2/N$.

\subsubsection{Star 3-uniform hypergraph}

We calculate the fixation probability for model 2 on the star 3-uniform hypergraph. As in the case of model 1, we exploit the fact that the combination of the type of the hub (i.e., either A or B) and the number of leaf nodes of type A, which ranges between $0$ and $N-1$, completely specifies the state of the evolutionary dynamics on this hypergraph. We show the derivation of the fixation probability in Text K in S1 File.

We obtain
\begin{equation}\label{eq:m2_star_N=4}
x_2=\frac{5r}{2(5r+3)}+\frac{3r}{2(3r+5)}
\end{equation}
for $N=4$ and
\begin{equation}\label{eq:m2_star_N=5}
x_2=\frac{6(r^3+3r^2)}{5(3r^3+13r^2+15r+4)}+\frac{3r^2}{5(r^2+3r+1)}
\end{equation}
for $N=5$. We show $x_3$ for $N = 5$ in Text A in S1 File.

We compare Eqs.~\eqref{eq:m2_star_N=4} and \eqref{eq:m2_star_N=5} with $x_2$ for the Moran process with $N=4$ and $N=5$ in
Fig~\ref{fig:m2_comparison}A and \ref{fig:m2_comparison}B, respectively. We find that the star 3-uniform hypergraph with $N=4$ is a suppressor of selection. In contrast, the star 3-uniform hypergraph with $N=5$ is neither an amplifier nor suppressor of selection because $x_2=9/35$ at $r=1$, which is different from the corresponding result for the Moran process, i.e., $x_2=2/5$. We also obtained $x_2$ for $N=20$ and $N=200$ by numerically solving a system of linear equations with $2N-2$ unknowns (see Text K in S1 File). Fig~\ref{fig:m2_comparison}C and \ref{fig:m2_comparison}D, which compare the obtained $x_2$ values with $x_2$ for the Moran process for $N=20$ and $N=200$, respectively, indicate that these larger star 3-uniform hypergraphs are neither amplifiers nor suppressors of selection. Therefore, we conclude that star 3-uniform hypergraphs are neither amplifiers nor suppressors of selection for $N \ge 5$ under the evolutionary dynamics described by model 2.

\subsection{Neutral drift}\label{neutral}

The discussion of amplifier and suppressor of selection requires $x_1 = 1/N$ (or $x_2 = 2/N$ in the case of our model 2) when $r=1$.
In this section, we prove the following theorem, which justifies such discussion for model 1 and not for model 2.
\begin{theorem}
Consider the neutral drift, i.e., $r=1$. When there are initially $i$ mutants (i.e., type A) that are distributed uniformly at random over the $N$ nodes, the fixation probability for the mutant is equal to $i/N$.
\end{theorem}

\begin{proof}
Consider model 1', which is defined by the same evolutionary dynamics as that for model 1 with the initial condition in which each node $i$ is occupied by a distinct neutral type $\tilde{A}_i$. Therefore, there are initially $N$ types, all of which have the constant fitness equal to 1. The evolutionary dynamics terminates when a single type fixates. We denote by $q_i$ the fixation probability for type $\tilde{A}_i$ under the aforementioned initial condition. It should be noted that $\sum_{i=1}^N q_i = 1$.

Now we consider the original model 1 with the initial condition in which there is just one mutant located at node $i$. Then, the fixation probability for the mutant is equal to $q_i$. This is because model 1' is reduced to model 1 with the present initial condition if we regard type $\tilde{A}_i$ as the mutant type and all the $N-1$ types $\tilde{A}_j$, where $j\neq i$, as the resident type. We recall that $x_1$ is the average of the fixation probability over the $N$ initial conditions in which there is just one mutant. Therefore, we obtain $x_1 = \sum_{i=1}^N q_i / N = 1/N$.

Next we consider model 1 with the initial condition in which there are two mutants whose locations are selected uniformly at random. Assume that the two mutants are initially located at nodes 1 and 2. The mutant type fixates with probability $q_1 + q_2$ because model 1' is reduced to model 1 with the present initial condition if we regard type $\tilde{A}_1$ and $\tilde{A}_2$ as the mutant type and all the other $N-2$ types $\tilde{A}_j$, where $j\neq 1, 2$, as the resident type. Therefore, we obtain
$x_2 = \sum_{i=1}^N \sum_{j=1}^{i-1}(q_i + q_j) / \binom{N}{2} = (N-1)\sum_{i=1}^N q_i / \binom{N}{2} = 2/N$.

Similarly, we obtain $x_3 = \sum_{i=1}^N \sum_{j=1}^{i-1} \sum_{k=1}^{j-1} (q_i + q_j + q_k) / \binom{N}{3} = \binom{N-1}{2} \sum_{i=1}^N q_i / \binom{N}{3} = 3/N$.
It is straightforward to verify $x_i = i/N$ for the other $i$ values.
\end{proof}

\begin{remark}
This theorem is known for the birth-death process on conventional networks with the proof being unchanged
\cite{DonnellyWelsh1983MPCPS,MasudaOhtsuki2009NewJPhys,Broom2010ProcRSocA}.
\end{remark}

\begin{remark}\label{rem:model2-neutral}
The theorem does not hold true for model 2. This is because, in model 2, whether the parent node $u$ propagates its type to the other nodes in the selected hyperedge $e$ containing $u$ depends on the other nodes belonging to $e$. As an example, consider the case in which $e$ contains three nodes, $u$, $v_1$, and $v_2$. Parent $u$ imposes its type on $v_1$ and $v_2$ only when either $v_1$ or $v_2$ has the same type as $u$. Otherwise, i.e., if both $v_1$ and $v_2$ are of the opposite type as $u$'s, then no state change occurs after $u$ is selected as parent and hyperedge $e$ is selected. Model 1' cannot handle this situation. In model 1', the parent disseminates its type to all the other nodes in the selected hyperedge $e$, as in model 1, regardless of the type of the other nodes belonging to $e$. Therefore, we cannot map model 1' to model 2. In fact, Fig~\ref{fig:m2_comparison} shows that $x_2 \neq 2/N$ in most cases, verifying that Theorem 1 does not hold true in general.
\end{remark}

\section{Numerical results for empirical hypergraphs\label{sec:numerical}}

In this section, we carry out numerical simulations for empirical hypergraphs. The present study is motivated by the result that most networks are amplifiers of selection under the birth-death process \cite{Hindersin2015PLOS,  Cuesta2018PlosOne, Allen2021PlosComputBiol}. In section~\ref{sub:model1}, we showed that three model hypergraphs are suppressors of selection under model 1. In section~\ref{sub:model2}, we argued that one cannot discuss whether the same hypergraphs are amplifiers or suppressors of selection under model 2 except for the complete, cyclic, and star 3-uniform hypergraphs with $N=4$ nodes. This is because model 2 does not respect $x_2 = 2/N$ in general (see Remark~\ref{rem:model2-neutral}). Therefore, we focus on model 1 in this section and examine whether empirical hypergraphs tend to be suppressors of selection, as is the case for the complete, cyclic, and star 3-uniform hypergraphs.

Empirical, or general, hypergraphs are distinct from the complete, cyclic, and star 3-uniform hypergraphs in two main aspects. First, in general, empirical hypergraphs do not have much symmetry that we can exploit to simplify the probability transition matrix $P$. Therefore, pursuing analytical solutions, which would involve the solution of a linear system with $2^N-2$ unknowns, is formidable unless $N$ is small. Second, empirical hypergraphs contain hyperedges of different sizes, whereas the 3-uniform hypergraphs only have hyperedges of size 3. Therefore, in this section, we numerically examine the stochastic evolutionary dynamics on empirical hypergraphs. In each time step, we select a parent node with the probability proportional to its fitness from all the $N$ nodes. Then, the parent propagates its type to all the other nodes in a hyperedge $e$ to which the parent belongs. We select $e$ uniformly at random from all the hyperedges to which the parent belongs regardless of the size of the hyperedge.

To calculate the fixation probability of a single mutant of type A for an arbitrary hypergraph and for each value of $r$, we run the birth-death process until type A or B fixates for each node $v$ initially occupied by type A. Note that all the $N-1$ nodes except $v$ are initially of type B. For each $v$, we run $3 \times 10^3$ simulations for $r\geq 1$ and $4 \times 10^4$ simulations for $r<1$. We use a substantially larger number of simulations for $r<1$ than $r\geq 1$ because the fixation probability is small when $r$ is small and therefore it can be estimated with a higher accuracy with more simulations.
For a given value of $r$, we estimate the fixation probability of type A as a fraction of the runs in which type A has fixated among the $3 \times 10^3 \times N$ or $4 \times 10^4 \times N$ simulations; the factor $N$ is due to the $N$ different choices of $v$.

We examine the fixation probability on four empirical hypergraphs. The corporate club membership hypergraph (corporate hypergraph for short) contains 25 nodes and 15 hyperedges~\cite{Faust1997SocNetw, Kunegis2013Konect}. Each node represents a corporate executive officer. Each hyperedge represents a social organization such as clubs and boards of which some officers are members. The Enron hypergraph is an email communication network and has 143 nodes and 10,885 hyperedges~\cite{Klimt2004ProcEurConf, Benson2018ProcNatAcad}. Each node represents an email address at Enron. Each hyperedge is comprised of all recipient addresses of an email. The Senate committee hypergraph (Senate hypergraph for short) has 282 nodes and 315 hyperedges. Each node is a member of the US Senate. Each hyperedge represents committee memberships \cite{CCA2017Senate, Chodrow2021SciAdv}. The high-school hypergraph is a social contact network with 327 nodes and 7,818 hyperedges. Each node is a student. Each hyperedge represents to a group of students that were in proximity of one another \cite{Chodrow2021SciAdv, Mastrandrea2015PLOS}. All the four hypergraphs are connected hypergraphs.

In Fig~\ref{fig:em-hypergraph}, we show the relationships between the fixation probability for a single mutant as a function of $r$, for the four empirical hypergraphs, one per panel. We find that all the hypergraphs are suppressors of selection. We also simulate the birth-death process on the weighted one-mode projection of each empirical hypergraph. We find that the fixation probability for the obtained weighted network is almost the same as that for the Moran process for the corporate (Fig~\ref{fig:em-hypergraph}A), Senate (Fig~\ref{fig:em-hypergraph}C), and high-school (Fig~\ref{fig:em-hypergraph}D) hypergraphs. The one-mode projection of the Enron hypergraph is a suppressor of selection (Fig~\ref{fig:em-hypergraph}B). However, the fixation probability for the one-mode projection of the Enron hypergraph as a function of $r$ is close to that for the Moran process. In contrast, the original Enron hypergraph is a much stronger suppressor of selection. Therefore, our result that the four empirical hypergraphs are suppressors of selection is not expected from the one-mode projection.

To further examine the result that the empirical hypergraphs are suppressors of selection, we now simulate the same evolutionary dynamics on the randomized hypergraph. We obtained the randomized hypergraph for each empirical hypergraph by randomly shuffling the hyperedges of the original hypergraph. In the random shuffling, we preserved the degree of each node and the size of each hyperedge \cite{Chodrow2020J.ComplexNetw, Nakajima2022IEEE}. We show the fixation probability for the randomized hypergraphs by the green circles in Fig~\ref{fig:em-hypergraph}. We find that the randomized hypergraph is also a suppressor of selection for all the four hypergraphs.
The randomized hypergraph is less suppressing than the original hypergraph for three empirical hypergraphs (see Fig~\ref{fig:em-hypergraph}A-C) and vice versa for the other one empirical hypergraph (see Fig~\ref{fig:em-hypergraph}D). Therefore, how the features of empirical hypergraphs except for the distribution of the node's degree and hyperedge's size affects the fixation probability is non-equivocal. However, these results for the randomized hypergraphs further strengthen our main claim that the hypergraphs that we have investigated are suppressors of selection under model 1.

\floatsetup[figure]{style=plain,subcapbesideposition=top}
\captionsetup{font={small,rm}} 
\captionsetup{labelfont=bf}
\begin{figure}[H]
  \centering
  \includegraphics[width=1.0\linewidth]{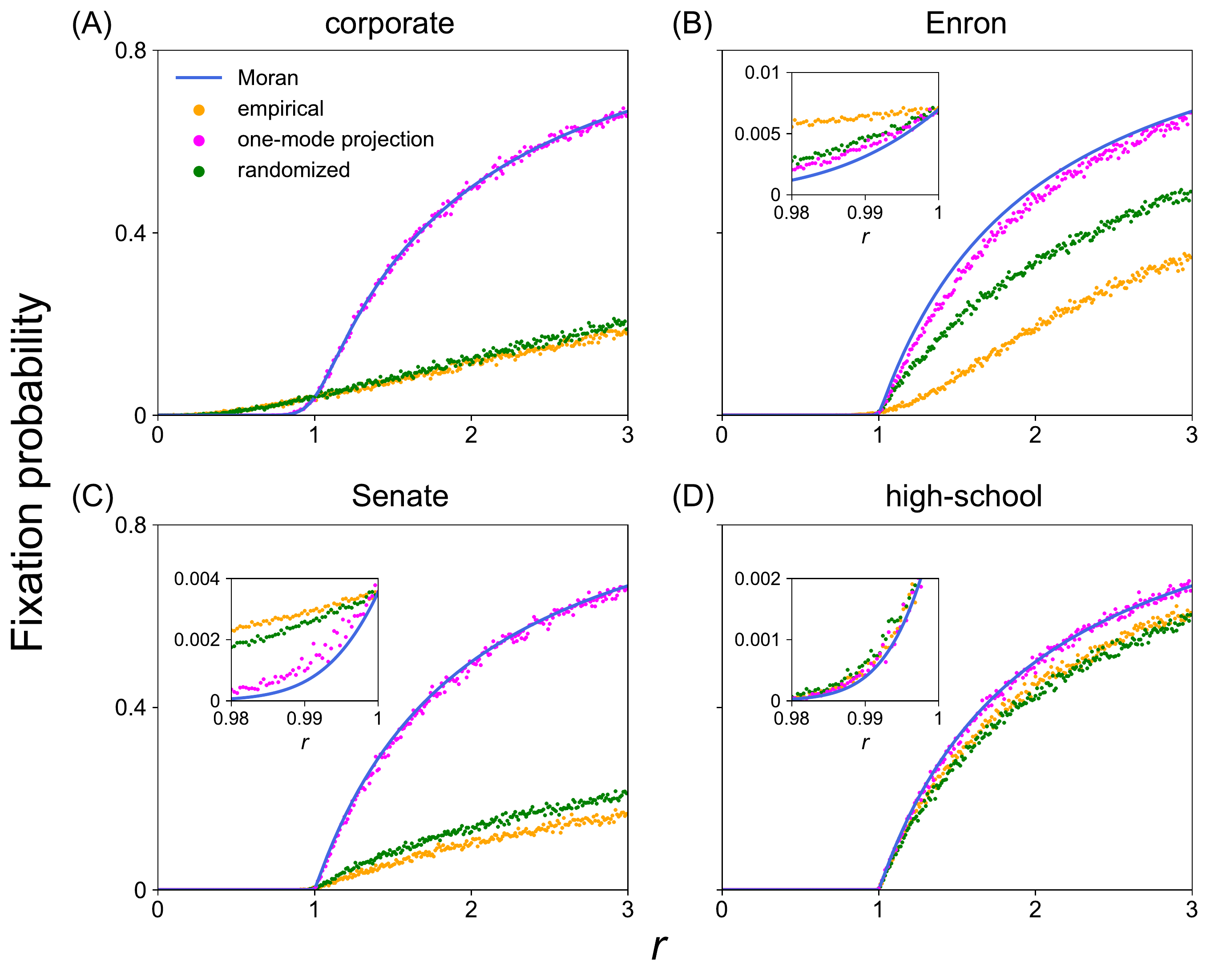}
\caption{Fixation probability for empirical hypergraphs, their one-mode projection, and the randomized hypergraphs. (A) Corporate. (B) Enron. (C) Senate. (D) High-school. The insets of (B), (C), and (D) magnify the results for $r$ values smaller than and close to $r=1$. We estimated the fixation probability at each value of $r$ as the fraction of runs in which fixation of type A is reached. We find that empirical and randomized hypergraphs are suppressors of selection for all the four data sets.
}
\label{fig:em-hypergraph}
\end{figure}

\section{Discussion}

We have proposed two models of evolutionary dynamics on hypergraphs that are extensions of the birth-death process on networks.
For both models of evolutionary dynamics, we semi-analytically derived the fixation probability for the mutant type under the constant selection for three synthetic hypergraphs with high symmetry, which generalize the complete graph, cycle graph, and star graph. For model 1, which is appropriate for discussing the strength of natural selection, we showed that these synthetic hypergraphs are suppressors of selection. Furthermore, by numerical simulations, we have shown that four empirical hypergraphs of our arbitrary choices are also suppressors of selection. Our results are in stark contrast to the known result that most networks are amplifiers of selection under the birth-death updating rule and the uniform initialization \cite{Hindersin2015PLOS, Cuesta2018PlosOne, Allen2021PlosComputBiol}, which we also assumed. It is often the case that interaction among individuals is often of higher order, and hypergraphs rather than conventional networks are more direct representation of many empirical data of social and ecological interactions \cite{Battiston2020PhyRep,Bianconi2021book}. Therefore, amplification of natural selection by birth-death processes on networks may not be as universal as it has been suggested once we expand the class of network models to be considered.

For conventional networks, finding suppressors of selection under the birth-death process is difficult \cite{Cuesta2017PlosOne}. In contrast, under the death-birth process, most conventional networks are suppressors of selection \cite{Hindersin2015PLOS}, and amplification of selection is bounded and transient~\cite{Tkadlec2020PlosComputBiol}. Our main result that all the hypergraphs that we have investigated are suppressors of selection under model 1 begs various related research questions. Is there a theoretical bound on amplification of selection in birth-death processes on hypergraphs? Is there a systematically constructed class of amplifying hypergraphs even if they are rare? Are hypergraphs suppressors of selection under the death-birth process? If it is the case, hypergraphs are even more suppressing under the death-birth than birth-death processes? Can we find an optimal amplifier \cite{Galanis2017JAcm,Pavlogiannis2018CommBiol,Goldberg2019TheorComputSci} or suppressor of selection for hypergraphs? We save these questions for the future.

As has been done in other studies, we used the Moran process as the baseline for judging the amplifying and suppressing effects of selection. An alternative is to use a hypergraph equivalent of the complete graph as the baseline. The complete $k$-uniform hypergraph is a candidate. 
Fig~\ref{fig:m1_comparison} and Text H in S1 File imply that, relative to the complete 3-uniform hypergraph and under model 1, the star 3-uniform hypergraphs with $N \in \{ 4, 5, 20, 200, 1500 \}$ are amplifiers of selection, and the cyclic 3-uniform hypergraph is neither an amplifier or suppressor of selection when $N \in \{5, 20, 200 \}$ (and equivalent to the complete 3-uniform hypergraph when $N=4$). Therefore, the cyclic 3-uniform hypergraph with $N \in \{5, 20, 200 \}$ is not an isothermal graph relative to the complete 3-uniform hypergraph, whereas the conventional circle graph is an isothermal graph relative to the Moran process. The result may change if we use a different baseline hypergraph such as the one that contains all the possible hyperedges whose size is at most $k$ or the complete $k$-uniform hypergraph with a different $k$ value. We may want to use a baseline satisfying that cyclic uniform hypergraphs are isothermal graphs relative to the baseline.
Our choice of the Moran process as the baseline is to avoid these difficulties.
How to determine the baseline hypergraph for discussing amplification and suppression of selection is an open question.

Metapopulation models assume networks in which a node is a patch containing individuals. Individuals either travel from a patch to an adjacent one to another according to a mobility rule such as a simple random walk or do not move but interact with the other individuals in the same patch more frequently than those in adjacent patches. Extending the island models in which the patches are connected as the complete graph \cite{Maruyama1970GenetRes}, constant-selection dynamics in metapopulation network models have been investigated \cite{Yagoobi2023JRSoc}. Metapopulation models and hypergraphs are apparently similar because the entire network is composed of interacting groups of individuals (i.e., patches or hyperedges) in both models. Recent studies showed that metapopulation models can be suppressors of selection \cite{Yagoobi2021SciRep,Marrec2021PRL}. However, constant-selection dynamics in metapopulation models and hypergraphs are crucially different in the following two aspects. First, interaction between individuals is pairwise in metapopulation models. In contrast, our hypergraph models assume group interaction, allowing multiple individuals to simultaneously change their type. Second, metapopulation models assume that within-group interaction between individuals is more likely than between-group interaction. This is not the case for hypergraphs in general although one can enforce it by, for example, allowing duplicated or weighted hyperedges representing a relatively high frequency of within-group interaction. Therefore, we consider that the hypergraphs used in the present study, and potentially many others, are suppressors of selection for different reasons from those for suppressing metapopulation networks, such as the possibility of multiple individuals to be simultaneously updated in the case of the hypergraph.

Evolutionary set theory is a mathematical framework with which to analyze coevolution of the strategy (i.e., type) of the node and the membership of the node to sets (i.e., groups) \cite{Tarnita2009PNAS,Nowak2010PhilTransRSocB,FuTarnita2012SciRep}. It assumes that each node belongs to different groups and play games with other nodes belonging to the same group. A node $v$ with a large fitness obtained from playing the game disseminates $v$'s group membership as well as $v$'s type to other nodes with a high probability. Therefore, evolutionary set theory is a dynamic graph theory. Evolutionary set theory is distinct from evolutionary dynamics on hypergraphs studied in the present study in the following aspects. First, in evolutionary set theory, interaction between players is pairwise by default. In contrast, in evolutionary dynamics on hypergraphs, the outcome of an interaction in a group (i.e., hyperedge) may not be decomposed into the summation of pairwise interaction in the group. For example, we assumed that a parent node $v$ simultaneously imposes its type on all the other nodes in the selected hyperedge. Second, the group membership evolves in evolutionary set theory. In contrast, it is fixed in the evolutionary dynamics considered in this study, which is a simplification assumption. Third, reproduction in evolutionary set theory occurs globally, not limited to between nodes in the same group. In contrast, we have assumed that reproduction occurs between nodes in the same hyperedge. These differences create opportunities for future work. For example, extending the present model to the case of dynamic hypergraphs may be interesting.

For conventional networks, the time to fixation has been shown to be in a tradeoff relationship with the fixation probability \cite{Frean2013ProcRSocB,Moller2019CommBiol,Tkadlec2019CommBiol, Tkadlec2021NatComm, Kuo2021biorxiv}. 
%
%
%
%
%
In other words, a strongly amplifying network tends to accompany a large mean fixation time. Our mathematical and computational framework is readily adaptable to the examination of fixation times. The fixation time for hypergraphs including the comparison with the case of conventional networks warrants future work. Other topics for further investigations include the effects of different initial conditions \cite{Antal2006PhysRevLett,MasudaOhtsuki2009NewJPhys,Adlam2015ProcRSocA,XiaoWu2019PlosComputBiol}, mathematical analyses of other symmetric hypergraphs, weak selection expansion of fixation probability \cite{Allen2021PlosComputBiol} in the case of hypergraphs, amplification and suppression of selection in the mutation-selection equilibrium under a small mutation rate \cite{SharmaTraulsen2022Pnas}, and fixation probability of cooperation in evolutionary game dynamics on hypergraphs~\cite{Santos2008nature, Rodriguez2021natHumBehav}. Focusing on fixation of a mutant (or resident) type in general allows us to use Markov chain theory to reach various mathematical insights. We believe that the present study is an important step toward deploying this powerful mathematical and computational machinery to evolutionary dynamics on higher-order networks.

\section*{Data Availability}

All data and code used for running simulations, model fitting, and plotting is available on a GitHub repository at: \url{https://github.com/RuodanL/fixation_probability}.

\section*{Supporting information}

\noindent\textbf{S1 Text. Fig A. Fixation probability for the star 3-uniform hypergraph with $N=1500$.} We compare it with the fixation probability for the Moran process. The inset on the left magnifies the result for $r$ values smaller than and close to $r=1$. The inset on the right magnifies the result for $r$ values greater than and close to $r=1$. In the main plot, the result for the Moran process, shown by the black line, is not identical but close to that for the star 3-uniform hypergraph, shown by the green line, such that the former is hidden behind the latter.
\textbf{Fig B. Fixation probability for the weighted one-mode projection of star 3-uniform hypergraphs.} We compare it with the fixation probability for the Moran process and star 3-uniform hypergraphs. (A) $N=4$. (B) $N=5$. (C) $N=20$. (D) $N=200$. The insets in (C) and (D) magnify the results for $r$ values smaller than and close to $r=1$. In the inset in (D), the results for the star 3-uniform hypergraph (shown by the blue line) and the Moran process (shown by the black line) are close to that for the one-mode projection (shown by the orange line) such that the blue and the black lines are almost hidden behind the orange line. In (A), (B), and the main panel of (D), the results for the Moran process (shown by the black lines) are not identical but close to those for the one-mode projection (shown by the orange lines) such that the former are hidden behind the latter. 
\textbf{Fig C. Initial position of the two nodes of type A on the cyclic 3-uniform hypergraph.} (A) The two nodes of type A do not share any hyperedge. (B) The two nodes of type A share two hyperedges. (C) The two nodes of type A share one hyperedge.
\textbf{Text A. Fixation probabilities for small 3-uniform hypergraphs.}
\textbf{Text B. Proof of Eq.~(22) for $N=4$.}
\textbf{Text C. Derivation of the entries of the transition probability matrix for $i \in \{2, \ldots, N-1\}$ for the cyclic 3-uniform hypergraph under model 1.}
\textbf{Text D. Proof of Eq.~(27) for $N=4$.}
\textbf{Text E. Proof of Eq.~(31) for $i=N-3$ and $i=N-2$.}
\textbf{Text F. Proof of Eq.~(32) for $i=2$ and $i=3$.}
\textbf{Text G. Derivation of the fixation probability for the star 3-uniform hypergraph under model 1.}
\textbf{Text H. Fixation probability for the star 3-uniform hypergraph with $N=1500$.}
\textbf{Text I. Fixation probability for the birth-death process on the weighted one-mode projection of the star 3-uniform hypergraph.}
\textbf{Text J. Derivation of the fixation probability for the cyclic 3-uniform hypergraph under model 2.}
\textbf{Text K. Derivation of the fixation probability for the star 3-uniform hypergraph under model 2.}

\noindent (PDF)

\section*{Author Contributions}

\noindent\textbf{Conceptualization:} Naoki Masuda.\\

\noindent\textbf{Formal analysis:} Ruodan Liu.\\

\noindent\textbf{Funding acquisition:} Naoki Masuda.\\

\noindent\textbf{Investigation:} Ruodan Liu, Naoki Masuda.\\

\noindent\textbf{Methodology:} Ruodan Liu, Naoki Masuda.\\

\noindent\textbf{Project administration:} Naoki Masuda.\\

\noindent\textbf{Software:} Ruodan Liu.\\

\noindent\textbf{Supervision:} Naoki Masuda.\\

\noindent\textbf{Validation:} Ruodan Liu.\\

\noindent\textbf{Visualization:} Ruodan Liu.\\

\noindent\textbf{Writing–original draft:} Ruodan Liu, Naoki Masuda.\\

\noindent\textbf{Writing–review\& editing:} Ruodan Liu, Naoki Masuda.

\newpage

\makeatletter
\newif\if@plaintag
\@plaintagfalse
\let\plaintagform@\tagform@
\newcommand{\plaintag}[1]{\begingroup\let\tagform@\plaintagform@\tag{#1}\endgroup}
\def\tagform@#1{\if@plaintag\else\refstepcounter{equation}\fi%
                \maketag@@@{(\textsc{\ignorespaces #1})}}
\newcommand{\ptag}[1]{\@plaintagtrue\tag{#1}\@plaintagfalse}
\makeatother

\section*{\centering{\Large{S1 Text}}}

\section*{Text A. Fixation probabilities for small 3-uniform hypergraphs}

\subsection*{A.1 Fixation probabilities for small 3-uniform hypergraphs under model 1}

For the complete 3-uniform hypergraph with $N=4$, we obtain
\begin{linenomath}
\begin{align}
x_2&=\frac{r(3r^2+8r+1)}{3(r+1)^3},\tag{S1}\\
x_3&=\frac{r(r+2)}{(r+1)^2}.\tag{S2}
\end{align}
\end{linenomath}
For the complete 3-uniform hypergraph with $N=5$, we obtain
\begin{linenomath}
\begin{align}
x_2&=\frac{r^2(8r^2+26r+8)}{8r^4+28r^3+33r^2+28r+8},\tag{S3}\\
x_3&=\frac{r(8r^3+28r^2+25r+2)}{8r^4+28r^3+33r^2+28r+8},\tag{S4}\\
x_4&=\frac{4r(2r^3+7r^2+8r+4)}{8r^4+28r^3+33r^2+28r+8}.\tag{S5}
\end{align}
\end{linenomath}
For the cyclic 3-uniform hypergraph with $N=4$, we obtain
\begin{linenomath}
\begin{align}
x_2&=\frac{r(3r^2+8r+1)}{3(r+1)^3},\tag{S6}\\
x_3&=\frac{r(r+2)}{(r+1)^2}.\tag{S7}
\end{align}
\end{linenomath}
For the cyclic 3-uniform hypergraph with $N=5$, we obtain
\begin{linenomath}
\begin{align}
x_2&=\frac{6r^2(r^2+3r+1)}{6r^4+20r^3+23r^2+20r+6},\tag{S8}\\
x_3&=\frac{r(6r^3+20r^2+17r+2)}{6r^4+20r^3+23r^2+20r+6},\tag{S9}\\
x_4&=\frac{2r(3r^3+10r^2+11r+6)}{6r^4+20r^3+23r^2+20r+6}.\tag{S10}
\end{align}
\end{linenomath}
For the star 3-uniform hypergraph with $N=4$, we obtain
\begin{linenomath}
\begin{align}
x_2&=\frac{r(3r^2+7r+1)}{3r^3+8r^2+8r+3},\tag{S11}\\
x_3&=\frac{3r(3r+5)}{4(3r^2+5r+3)}+\frac{9r(r^2+2r+1)}{4(9r^3+18r^2+14r+3)}.\tag{S12}
\end{align}
\end{linenomath}
For the star 3-uniform hypergraph with $N=5$, we obtain
\begin{linenomath}
\begin{align}
x_2&=\frac{2r^2(144r^3+548r^2+512r+71)}{5(144r^5+620r^4+1040r^3+1055r^2+750r+216)}+\frac{3r^2(72r^2+196r+72)}{5(72r^4+202r^3+217r^2+202r+72)},\tag{S13}\\
x_3&=\frac{3r(72r^3+202r^2+145r+6)}{5(72r^4+202r^3+217r^2+202r+72)}+\frac{12r(36r^4+125r^3+164r^2+88r+12)}{5(216r^5+750r^4+1055r^3+1040r^2+620r+144)},\tag{S14}\\
x_4&=\frac{4r(72r^3+202r^2+213r+108)}{5(72r^4+202r^3+217r^2+202r+72)}+\frac{8r(36r^4+110r^3+133r^2+110r+36)}{5(288r^5+880r^4+1070r^3+1025r^2+490r+72)}.\tag{S15}
\end{align}
\end{linenomath}

\subsection*{A.2 Fixation probabilities for small 3-uniform hypergraphs under model 2}

We obtain $x_1 = 0$ and $x_{N-1}=1$ for any 3-uniform hypergraph under model 2. Therefore, when $N=4$, only $x_2$, which we show in the main text, is nontrivial. When $N=5$, only $x_2$ and $x_3$ are nontrivial. Because we show $x_2$ in the case of $N=5$ in the main text, here we show $x_3$ in the case of $N=5$. For the cyclic 3-uniform hypergraph with $N=5$, we obtain
\begin{linenomath}
\begin{equation}
x_3=\frac{r^2-1}{2(r^2-r^{-1})}+\frac{1}{2(1+4r)}\left(4r+\frac{r-1}{r-r^{-2}}\right).\tag{S16}
\end{equation}
\end{linenomath}
For the star 3-uniform hypergraph with $N=5$, we obtain
\begin{linenomath}
\begin{equation}
x_3=\frac{3r(r+3)}{5(r^2+3r+1)}+\frac{2r(4r^2+15r+4)}{5(4r^3+15r^2+13r+3)}.\tag{S17}
\end{equation}
\end{linenomath}

\section*{Text B. Proof of Eq.~(22) for $N=4$\label{sec:model1-cyclic-case1-N=4-proof}}


If $N=4$, then $v_{\ell-2}$ and $v_{\ell+2}$ are identical. In this case, there are two sequences of events through which the state moves from $i=1$ to $i=0$ in one time step.
In the first sequence, $v_{\ell-2}$, which is of type B, is selected as parent with probability $1/(r+3)$.
Then, the hyperedge that contains the parent and $v_{\ell}$, i.e., $\{\ell-2, \ell-1, \ell\}$ or $\{\ell, \ell+1, \ell-2\}$, is used for reproduction, which occurs with probability $2/3$. 
In the second sequence, either $v_{\ell-1}$ or $v_{\ell+1}$ is selected as parent, which occurs with probability $2/(r+3)$. Then, one of the two hyperedges that contain the parent and $v_{\ell}$ is used for reproduction, which occurs with probability $2/3$. For example, if $v_{\ell-1}$ is selected as parent and hyperedge $\{ \ell-2, \ell-1, \ell \}$ or $\{ \ell-1, \ell, \ell+1 \}$ is used for reproduction, then the state moves from $1$ to $0$.
By summing up these probabilities, we obtain Eq.~(22).

\section*{Text C. Derivation of the entries of the transition probability matrix for $i \in \{2, \ldots, N-1\}$ for the cyclic 3-uniform hypergraph under model 1\label{sec:cyclic-model1-derivation}}

In this section, we derive the $(i, j)$ entries of the transition probability matrix for the cyclic 3-uniform hypergraph under model 1, where
$i \in \{2, \ldots, N-1\}$.

Similarly to the case of $i=1$, when $i=N-1$, there are three types of events that can occur in a time step. Without loss of generality, we assume that the $\ell$th node is the only node of type B (see Fig~3B for a schematic).
The state moves from $i=N-1$ to $i=N-3$ whenever $v_\ell$ is selected as parent, which occurs with probability $1/[r(N-1)+1]$. Therefore, we obtain
\begin{equation}
p_{N-1,N-3}=\frac{1}{r(N-1)+1}\tag{S18},
\end{equation}
which is Eq.~(26) in the main text.
Alternatively, the state moves from $i=N-1$ to $i=N$ such that type A fixates in the following two cases.
In the first case, either node $v_{\ell-2}$ or $v_{\ell+2}$, which is of type A, is selected as parent. If $N\ge 5$, these two nodes are distinct. Therefore, this event occurs with probability $2r/[r(N-1)+1]$.
Then, the hyperedge that contains the parent and $v_{\ell}$ is used for reproduction, which occurs with probability $1/3$.
%
%
In the second case, either $v_{\ell-1}$ or $v_{\ell+1}$ is selected as parent, which occurs with probability $2r/[r(N-1)+1]$. Then, one of the two hyperedges that contain the parent and $v_{\ell}$ is used for reproduction, which occurs with probability $2/3$. 
%
%
By summing up these probabilities, we obtain
\begin{equation}
p_{N-1,N}=\frac{2r}{r(N-1)+1}\tag{S19},
\label{eq:model1-cycle-i=n-1-increase1-SI}
\end{equation}
which is Eq.~(27) in the main text.
In fact, Eq.~\eqref{eq:model1-cycle-i=n-1-increase1-SI} also holds true for $N=4$; see~Text D for the proof.
If any other event occurs, then $i=N-1$ remains unchanged. Therefore, we obtain
\begin{linenomath}
\begin{align}
p_{N-1,N-1}&=1-p_{N-1,N-3}-p_{N-1,N}=\frac{r(N-3)}{r(N-1)+1}\tag{S20},\\
p_{N-1, j}&=0 \text{ if } j \neq N-3 , N-1, N,\tag{S21}
\end{align}
\end{linenomath}
which are Eqs.~(28) and (29), respectively, in the main text.

When $i\in \{2, \ldots, N-2\}$, there are five types of possible events. Without loss of generality, we assume that the $\ell$th to the $(\ell+i-1)$th nodes are of type A and that all other nodes are of type B (see Fig~3C for a schematic). If $\ell+i-1$ is larger than $N$, then we interpret $\ell+i-1$ as the number modulo $N$ (i.e., $\ell+i-1-N$); the same convention applies in the following text.
In the first type of event, either node $v_{\ell-1}$ or $v_{\ell+i}$, which is of type B, is selected as parent; this event occurs with probability $2/(ri+N-i)$. Then, the hyperedge that contains the parent and two nodes of type A is used for reproduction, which occurs with probability $1/3$. In this case, the state $i$ decreases by two. For example, if $v_{\ell-1}$ is selected as parent and hyperedge $\{ \ell-1, \ell, \ell+1 \}$ is selected, then the state moves from $i$ to $i-2$. Therefore, we obtain
\begin{equation}
p_{i,i-2}=\frac{2}{3(ri+N-i)},\tag{S22}
\end{equation}
which is Eq.~(30) in the main text.
In the second type of event, either $v_{\ell-2}$, $v_{\ell-1}$, $v_{\ell+i}$, or $v_{\ell+i+1}$, which is of type B, is selected as parent. If $i\le N-4$, these four nodes are distinct. Therefore, this event occurs with probability $4/(ri+N-i)$. Then, the hyperedge that contains the parent, a node of type B, and a node of type A is used for reproduction, which occurs with probability $1/3$. 
In this case, the state $i$ decreases by one. For example, if $v_{\ell-2}$ is selected as parent and hyperedge $\{ \ell-2, \ell-1, \ell \}$ is used for reproduction, then the state moves from $i$ to $i-1$ because $v_{\ell}$ turns from A to B. Therefore, we obtain
\begin{equation}
p_{i,i-1}=\frac{4}{3(ri+N-i)},\tag{S23}
\label{eq:model1-cyclic-decrease1-SI}
\end{equation}
which is Eq.~(31) in the main text.
In fact, Eq.~(eqref{eq:model1-cyclic-decrease1-SI} also holds true for $i=N-3$ and $i=N-2$ although some of $v_{\ell-2}$, $v_{\ell-1}$, $v_{\ell+i}$, and $v_{\ell+i+1}$ are identical nodes when $i=N-3$ or $i=N-2$; see~Text E for the proof.
In the third type of event, either $v_\ell$, $v_{\ell+1}$, $v_{\ell+i-2}$, or $v_{\ell+i-1}$, which is of type A, is selected as parent. If $i\ge 4$, these four nodes are distinct. Therefore, this event occurs with probability $4r/(ri+N-i)$. Then, the hyperedge that contains the parent node, a node of type A, and a node of type B is used for reproduction, which occurs with probability $1/3$.
In this case, state $i$ increases by one. For example, if $v_\ell$ is selected as parent and hyperedge $\{ \ell-1, \ell, \ell+1 \}$ is used for reproduction, then the state moves from $i$ to $i+1$ because $v_{\ell-1}$ turns from B to A. Therefore, we obtain
\begin{equation}
p_{i,i+1}=\frac{4r}{3(ri+N-i)},\tag{S24}
\label{eq:model1-cyclic-increase1-SI}
\end{equation}
which is Eq.~(32) in the main text.
In fact, Eq.~\eqref{eq:model1-cyclic-increase1-SI} also holds true for $i=2$ and $i=3$ although some of $v_\ell$, $v_{\ell+1}$, $v_{\ell+i-2}$, and $v_{\ell+i-1}$ are identical nodes when $i=2$ or $i=3$; see~Text F for the proof. 
In the fourth type of event, either node $v_{\ell}$ or $v_{\ell+i-1}$, which is of type A, is selected as parent; this event occurs with probability $2r/(ri+N-i)$. Then, the hyperedge that contains the parent and two nodes of type B is used for reproduction, which occurs with probability $1/3$. In this case, state $i$ increases by two. For example, if $v_{\ell}$ is selected as parent and hyperedge $\{ \ell-2, \ell-1, \ell \}$ is used for reproduction, then the state moves from $i$ to $i+2$. Therefore, we obtain
\begin{equation}
p_{i,i+2}=\frac{2r}{3(ri+N-i)},\tag{S25}
\end{equation}
which is Eq.~(33) in the main text.
If any other event occurs, then $i$ remains unchanged. Therefore, we obtain
\begin{equation}
p_{i,i}=1-p_{i,i-2}-p_{i,i-1}-p_{i,i+1}-p_{i,i+2},\tag{S26}
\end{equation}
which is Eq.~(34) in the main text.

\section*{Text D. Proof of Eq.~(27) for $N=4$\label{sec:model1-cyclic-case2-N=4-proof}}

If $N=4$, then $v_{\ell-2}$ and $v_{\ell+2}$ are identical. In this case, there are two sequences of events through which the state moves from $i=N-1$ to $i=N$.
In the first sequence, $v_{\ell-2}$, which is of type A, is selected as parent with probability $r/(3r+1)$.
Then, the hyperedge that contains the parent and $v_{\ell}$, i.e., $\{\ell-2, \ell-1, \ell\}$ or $\{\ell, \ell+1, \ell-2\}$, is used for reproduction, which occurs with probability $2/3$. 
In the second sequence, either $v_{\ell-1}$ or $v_{\ell+1}$, which is of type A, is selected as parent with probability $2r/(3r+1)$. Then, one of the two hyperedges that contain the parent and $v_{\ell}$ is used for reproduction, which occurs with probability $2/3$. 
By summing up these probabilities, we obtain Eq.~(27).

\section*{Text E. Proof of Eq.~(31) for $i=N-3$ and $i=N-2$\label{sec:model1-cyclic-case2-i=n-3and-n-2-proof}}

If $i=N-3$, then $v_{\ell-2}$ and $v_{\ell+i+1}$ are identical. In this case, there are two sequences of events through which the state decreases from $i$ to $i-1$. In the first sequence, either $v_{\ell-1}$ or $v_{\ell-3}$ is selected as parent, which occurs with probability $2/(rN-3r+3)$. Then, the hyperedge containing two nodes of type B (i.e., hyperedge $\{\ell-2, \ell-1, \ell\}$ if $v_{\ell-1}$ is the parent and hyperedge $\{\ell-4, \ell-3, \ell-2\}$ if $v_{\ell-3}$ is the parent) is used for reproduction, which occurs with probability $1/3$. In the second sequence, $v_{\ell-2}$ is selected as parent with probability $1/(rN-3r+3)$. Then, hyperedge $\{\ell-2, \ell-1, \ell\}$ or $\{\ell-4, \ell-3, \ell-2\}$ is used for reproduction, which occurs with probability $2/3$. By summing up these probabilities, we obtain Eq.~(31). 

If $i=N-2$, either of the two nodes of type B, i.e., $v_{\ell-1}$ or $v_{\ell-2}$, must be selected as parent for the state to move from $i$ to $i-1$. This event occurs with probability $2/(rN-2r+2)$. Then, either hyperedge $\{\ell-2, \ell-1, \ell\}$ or $\{\ell-3, \ell-2, \ell-1\}$ must be used for reproduction, which occurs with probability $2/3$. The product of these two probabilities yields Eq.~(31).

\section*{Text F. Proof of Eq.~(32) for $i=2$ and $i=3$\label{sec:model1-cyclic-case3-i=2and3-proof}}

If $i=2$, for the state to move from $i$ to $i+1$, either $v_\ell$ or $v_{\ell+1}$, which is of type A, must be selected as parent. This event occurs with probability $2r/(2r+N-2)$. Then, a hyperedge containing both $v_{\ell}$ and $v_{\ell+1}$ must be selected, which occurs with probability $2/3$. The product of these two probabilities yields Eq.~(32).

If $i=3$, then $v_{\ell+1}$ and $v_{\ell+i-2}$ are identical. In this case, there are two sequences of events with which the state increases from $i$ to $i+1$. In the first sequence, either $v_\ell$ or $v_{\ell+2}$ is selected as parent with probability $2r/(3r+N-3)$. Then, the hyperedge composed of the parent, $v_{\ell + 1}$, which is of type A, and a node of type B (i.e., $v_{\ell-1}$ if the parent is $v_{\ell}$, and $v_{\ell+3}$ if the parent is $v_{\ell+2}$) is used for reproduction with probability $1/3$. In the second sequence, $v_{\ell+1}$ is selected as parent with probability $r/(3r+N-3)$. Then, hyperedge \{${\ell-1}$, ${\ell}$, ${\ell+1}$\} or \{${\ell+1}$, ${\ell+2}$, ${\ell+3}$\} is used for reproduction with probability $2/3$. By summing up these probabilities, we obtain Eq.~(32).

\section*{Text G. Derivation of the fixation probability for the star 3-uniform hypergraph under model 1}\label{sub:star-model1-derivation}

In this section, we derive the fixation probability for the star 3-uniform hypergraph under model 1.

Assume that there are currently $i$ nodes of type A and that the state is $(i_1, i_2) = (1, i-1)$ with $i \in \{1, \ldots, N-1 \}$. There are five types of events that can occur next.

In the first type of event, a leaf node of type B is selected as parent, which occurs with probability $(N-i)/(ri+N-i)$. Then, a hyperedge that contains the parent, the hub node, and a different leaf node of type B is used for reproduction, which occurs with probability $(N-i-1)/(N-2)$. The state after this entire event is $(i_1, i_2) = (0, i-1)$ with $i \in \{1, \ldots, N-2 \}$. Therefore, we obtain
\begin{equation}
p_{(1, i-1) \to (0, i-1)} = \frac{N-i}{ri+N-i}\frac{N-i-1}{N-2}.\tag{S27}
\end{equation}
In the second type of event, a leaf node of type B is selected as parent with probability $(N-i)/(ri+N-i)$. Then, a hyperedge that contains the parent, the hub node, and a leaf node of type A is used for reproduction, which occurs with probability $(i-1)/(N-2)$. The state after this event is $(0, i-2)$ with $i \in \{2, \ldots, N-1 \}$. Therefore, we obtain
\begin{equation}
p_{(1, i-1) \to (0, i-2)} = \frac{N-i}{ri+N-i}\frac{i-1}{N-2}.\tag{S28}
\end{equation}
In the third type of event, the hub node, which is of type A, is selected as parent, which occurs with probability $r/(ri+N-i)$. Then, a hyperedge that contains the parent, a leaf node of type A, and a leaf node of type B is used for reproduction, which occurs with probability $(i-1)(N-i)/\binom{N-1}{2}$. Alternatively, a leaf node of type A is selected as parent, which occurs with probability $r(i-1)/(ri+N-i)$. Then, a hyperedge that contains the parent, the hub node, and a leaf node of type B is used for reproduction, which occurs with probability $(N-i)/(N-2)$. In both cases, the state after the event is $(1, i)$ with $i \in \{2, \ldots, N-1 \}$. Therefore, we obtain
\begin{equation}
p_{(1, i-1) \to (1, i)} = \frac{r}{ri+N-i}\frac{(i-1)(N-i)}{\binom{N-1}{2}}+\frac{r(i-1)}{ri+N-i}\frac{N-i}{N-2}.\tag{S29}
\end{equation}
In the fourth type of event, the hub node is selected as parent with probability $r/(ri+N-i)$. Then, a hyperedge that contains the parent and two leaf nodes of type B is used for reproduction, which occurs with probability $\binom{N-i}{2}/\binom{N-1}{2}$. The state after this event is $(1, i+1)$ with $i \in \{1, \ldots, N-2 \}$. Therefore, we obtain
\begin{equation}
p_{(1, i-1) \to (1, i+1)} = \frac{r}{ri+N-i}\frac{\binom{N-i}{2}}{\binom{N-1}{2}}.\tag{S30}
\end{equation}
If any other event occurs, then the state remains unchanged. Therefore, we obtain
\begin{equation}
p_{(1, i-1) \to (1, i-1)} = 1 - p_{(1, i-1) \to (0, i-1)}
- p_{(1, i-1) \to (0, i-2)} - p_{(1, i-1) \to (1, i)}
- p_{(1, i-1) \to (1, i+1)}.\tag{S31}
\end{equation}
We recall $\tilde{x}_{(i_1,i_2)}$ is the probability that type A fixates starting with state $(i_1, i_2)$.
We obtain
\begin{linenomath}
\begin{align}
\tilde{x}_{(1, i-1)} =& p_{(1, i-1) \to (0, i-1)} \tilde{x}_{(0, i-1)} + p_{(1, i-1) \to (0, i-2)} \tilde{x}_{(0, i-2)} + p_{(1, i-1) \to (1, i)} \tilde{x}_{(1, i)} \nonumber\\
&+ p_{(1, i-1) \to (1, i+1)}\tilde{x}_{(1, i+1)} + p_{(1, i-1) \to (1, i-1)} \tilde{x}_{(1, i-1)}.\tag{S32}
\label{eq:star-rho1}
\end{align}
\end{linenomath}

Assume that the current state is $(i_1, i_2) = (0, i)$ with $i \in \{1, \ldots, N-1\}$. There are five types of events that can occur next.

In the first type of event, a leaf node of type A is selected as parent with probability $ri/(ri+N-i)$. Then, a hyperedge that contains the parent, the hub node, and a different leaf node of type A is used for reproduction with probability $(i-1)/(N-2)$. The state after this entire event is $(i_1,i_2)=(1, i)$ with $i \in \{2, \ldots, N-1 \}$. Therefore, we obtain
\begin{equation}
p_{(0, i) \to (1, i)} = \frac{ri}{ri+N-i}\frac{i-1}{N-2}.\tag{S33}
\end{equation}
In the second type of event, a leaf node of type A is selected as parent with probability $ri/(ri+N-i)$. Then, a hyperedge that contains the parent, the hub node, and a leaf node of type B is used for reproduction, which occurs with probability $(N-i-1)/(N-2)$. The state after this event is $(1, i+1)$ with $i \in \{1, \ldots, N-2 \}$. Therefore, we obtain
\begin{equation}
p_{(0, i) \to (1, i+1)} = \frac{ri}{ri+N-i}\frac{N-i-1}{N-2}.\tag{S34}
\end{equation}
In the third type of event, the hub node, which is of type B, is selected as parent, with probability $1/(ri+N-i)$. Then, a hyperedge that contains the parent, a leaf node of type A, and a leaf node of type B, is used for reproduction, which occurs with probability $i(N-i-1)/\binom{N-1}{2}$. Alternatively, a leaf node of type B is selected as parent with probability $(N-i-1)/(ri+N-i)$. Then, a hyperedge that contains the parent, the hub node, and a leaf node of type A is used for reproduction, which occurs with probability $i/(N-2)$. In both cases, the state after the event is $(0, i-1)$ with $i \in \{1, \ldots, N-2 \}$. Therefore, we obtain
\begin{equation}
p_{(0, i) \to (0, i-1)} = \frac{1}{ri+N-i}\frac{i(N-i-1)}{\binom{N-1}{2}}+\frac{N-i-1}{ri+N-i}\frac{i}{N-2}.\tag{S35}
\end{equation}
In the fourth type of event, the hub node is selected as parent with probability $1/(ri+N-i)$. Then, the hyperedge that contains the parent and two leaf nodes of type A is used for reproduction, which occurs with probability $\binom{i}{2}/\binom{N-1}{2}$. The state after this event is $(0, i-2)$ with $i \in \{2, \ldots, N-1 \}$. Therefore, we obtain
\begin{equation}
p_{(0, i) \to (0, i-2)} = \frac{1}{ri+N-i}\frac{\binom{i}{2}}{\binom{N-1}{2}}.\tag{S36}
\end{equation}
If any other event occurs, then the state remains unchanged. Therefore, we obtain
\begin{linenomath}
\begin{align}
p_{(0, i) \to (0, i)} = 1 - p_{(0, i) \to (1, i)} - p_{(0, i) \to (1, i+1)} - p_{(0, i) \to (0, i-1)} - p_{(0, i) \to (0, i-2)}.\tag{S37}
\end{align}
\end{linenomath}
Using these transition probabilities, we obtain
\begin{linenomath}
\begin{align}
\tilde{x}_{(0, i)} =& p_{(0, i) \to (1, i)} \tilde{x}_{(1, i)} + p_{(0, i) \to (1, i+1)} \tilde{x}_{(1, i+1)} + 
p_{(0, i) \to (0, i-1)} \tilde{x}_{(0, i-1)} \nonumber\\
&+ p_{(0, i) \to (0, i-2)} \tilde{x}_{(0, i-2)} + p_{(0, i) \to (0, i)} \tilde{x}_{(0, i)}.\tag{S38}
\label{eq:star-rho2}
\end{align}
\end{linenomath}

We rewrite Eqs.~\eqref{eq:star-rho1} and \eqref{eq:star-rho2} as
\begin{equation}\label{veceqstar-SI}
\tilde{\bm{x}}=P\tilde{\bm{x}},\tag{S39}
\end{equation}
which is Eq.~(37) in the main text, where
\begin{equation}
P=
\left(\begin{array}{ c | c }
    C & D \\
    \hline
    E & F
\end{array}\right),\tag{S40}
\label{eq:P-block-SI}
\end{equation}
which is Eq.~(38) in the main text,
\begin{equation}
C=
\begin{pmatrix}
1 & 0 & 0 & \cdots & 0 & 0 & 0\\
p_{(0,1)\rightarrow(0,0)} & p_{(0,1)\rightarrow(0,1)} & 0 & \cdots & 0 & 0 & 0\\
p_{(0,2)\rightarrow(0,0)} & p_{(0,2)\rightarrow(0,1)} & p_{(0,2)\rightarrow(0,2)} & \cdots & 0 & 0 & 0\\
\vdots & \vdots & \vdots & \cdots & \vdots & \vdots & \\
0 & 0 & 0 & \cdots & p_{(0,N-1)\rightarrow(0,N-3)} & p_{(0,N-1)\rightarrow(0,N-2)} & p_{(0,N-1)\rightarrow(0,N-1)}
\end{pmatrix},\tag{S41}
\end{equation}
\begin{equation}
D=
\begin{pmatrix}
0 & 0 & 0 & 0 & \cdots & 0 & 0 & 0\\
0 & p_{(0,1)\rightarrow(1,1)} & p_{(0,1)\rightarrow(1,2)} & 0 & \cdots & 0 & 0 & 0\\
0 & 0 & p_{(0,2)\rightarrow(1,2)} & p_{(0,2)\rightarrow(1,3)} & \cdots & 0 & 0 & 0\\
\vdots & \vdots & \vdots & \vdots & \cdots & \vdots & \vdots & \\
0 & 0 & 0 & 0 & \cdots & 0 & p_{(0,N-2)\rightarrow(1,N-2)} & p_{(0,N-2)\rightarrow(1,N-1)}\\
0 & 0 & 0 & 0 & \cdots & 0 & 0 & p_{(0,N-1)\rightarrow(1,N-1)}
\end{pmatrix},\tag{S42}
\end{equation}
\begin{equation}
E=
\begin{pmatrix}
p_{(1,0)\rightarrow(0,0)} & 0 & 0 & \cdots & 0 & 0 & 0\\
p_{(1,1)\rightarrow(0,0)} & p_{(1,1)\rightarrow(0,1)} & 0 & \cdots & 0 & 0 & 0\\
0 & p_{(1,2)\rightarrow(0,1)} & p_{(1,2)\rightarrow(0,2)} & \cdots & 0 & 0 & 0\\
\vdots & \vdots & \vdots & \cdots & \vdots & \vdots & \vdots &\\
0 & 0 & 0 & \cdots & p_{(1,N-2)\rightarrow(0,N-3)} & p_{(1,N-2)\rightarrow(0,N-2)} & 0\\
0 & 0 & 0 & \cdots & 0 & p_{(1,N-1)\rightarrow(0,N-2)} & p_{(1,N-1)\rightarrow(0,N-1)}
\end{pmatrix},\tag{S43}
\end{equation}
and
\begin{equation}
F=
\begin{pmatrix}
p_{(1,0)\rightarrow(1,0)} & p_{(1,0)\rightarrow(1,1)} & p_{(1,0)\rightarrow(1,2)} & 0 & \cdots & 0 & 0 & 0\\
0 & p_{(1,1)\rightarrow(1,1)} & p_{(1,1)\rightarrow(1,2)} & p_{(1,1)\rightarrow(1,3)} & \cdots & 0 & 0 & 0\\
\vdots & \vdots & \vdots & \vdots & \cdots & \vdots & \vdots & \vdots &\\
0 & 0 & 0 & 0 & \cdots & 0 & p_{(1,N-2)\rightarrow(1,N-2)} & p_{(1,N-2)\rightarrow(1,N-1)}\\
0 & 0 & 0 & 0 & \cdots & 0 & 0 & 1
\end{pmatrix}.\tag{S44}
\end{equation}

\section*{Text H. Fixation probability for the star 3-uniform hypergraph with $N=1500$}

In Fig~\ref{fig:star_1500}, we compare the fixation probability for the star 3-uniform hypergraph having $N=1500$ nodes with that for the Moran process under model 1. The figure indicates that the star 3-uniform hypergraph with $N=1500$ is a suppressor of selection.
\renewcommand{\thefigure}{A}
\floatsetup[figure]{style=plain,subcapbesideposition=top}
\captionsetup{font={small,rm}} 
\captionsetup{labelfont=bf}
\begin{figure}[H]
  \centering
  \includegraphics[width=0.65\linewidth]{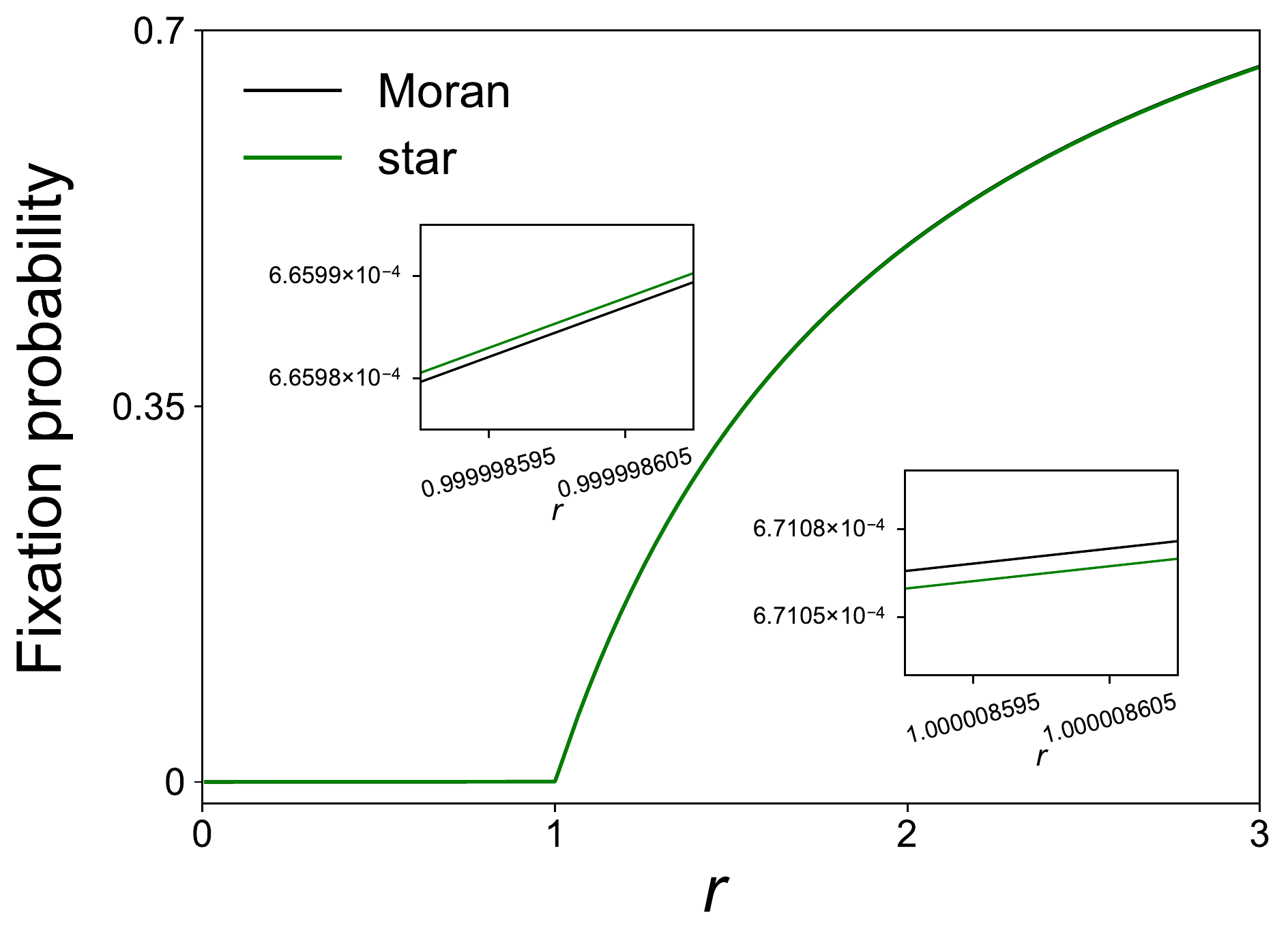}
   \caption{Fixation probability for the star 3-uniform hypergraph with $N=1500$. We compare it with the fixation probability for the Moran process. The inset on the left magnifies the result for $r$ values smaller than and close to $r=1$. The inset on the right magnifies the result for $r$ values greater than and close to $r=1$. In the main plot, the result for the Moran process, shown by the black line, is not identical but close to that for the star 3-uniform hypergraph, shown by the green line, such that the former is hidden behind the latter.}
   \label{fig:star_1500}
\end{figure}

\section*{Text I. Fixation probability for the birth-death process on the weighted one-mode projection of the star 3-uniform hypergraph\label{sec:one-mode-star}}

In this section, we consider the weighted one-mode projection of the star 3-uniform hypergraph and examine the fixation probability on the obtained weighted network. 

We denote the weighted one-mode projection of the star 3-uniform hypergraph by $G$. Note that $G$ is a weighted complete graph; the edge between the hub node and any leaf node has weight $N-2$, and the edge between any pair of leaf nodes has weight 1. In the birth-death process on $G$, in each time step, we select one node as parent with the probability proportional to its fitness. Then, the parent selects one of its neighbors with the probability proportional to the edge weight and converts the neighbor into the parent's type. As is the case for the star 3-uniform hypergraph, the symmetry in $G$ allows us to specify the state of the birth-death process by tuple $(i_1, i_2)$, where $i_1 \in \{ 0, 1 \}$ specifies whether the hub is of type A (i.e., $i_1 = 1$) or B (i.e., $i_1 = 0$), and $i_2 \in \{0, 1, \ldots, N-1 \}$ is the number of leaf nodes of type A. The total number of nodes of type A is equal to $i=i_1+i_2$. The fixation of type A and B corresponds to $(i_1, i_2) = (1, N-1)$ and $(0, 0)$, respectively.

Assume that the current state is $(i_1, i_2) = (1, i-1)$ with $i-1 \in \{0, 1, \ldots, N-2 \}$. There are four types of events that can occur in the next time step.
In the first type of event, a leaf node of type B is selected as parent with probability $(N-i)/(ri+N-i)$. Then, the edge between the parent and the hub node is used for reproduction with probability $1/2$. The state after this entire event is $(i_1, i_2) = (0, i-1)$. Therefore, we obtain
\begin{equation}
p_{(1, i-1) \to (0, i-1)} = \frac{N-i}{ri+N-i} \cdot \frac{1}{2}.\tag{S45}
\end{equation}
In the second type of event, a leaf node of type B is selected as parent with probability $(N-i)/(ri+N-i)$. Then, the edge between the parent and a leaf node of type A is used for reproduction with probability $(i-1)/2(N-2)$. The state after this event is $(1, i-2)$. Therefore, we obtain
\begin{equation}
p_{(1, i-1) \to (1, i-2)} = \frac{N-i}{ri+N-i} \cdot \frac{i-1}{2(N-2)}.\tag{S46}
\end{equation} 
In the third type of event, the hub node, which is of type A, is selected as parent with probability $r/(ri+N-i)$. Then, the edge between the parent and a leaf node of type B is used for reproduction with probability $(N-i)/(N-1)$. Alternatively, a leaf node of type A is selected as parent with probability $r(i-1)/(ri+N-i)$. Then, the edge between the parent and a leaf node of type B is used for reproduction with probability $(N-i)/2(N-2)$. In both cases, the state after the event is $(1, i)$. Therefore, we obtain
\begin{equation}
p_{(1, i-1) \to (1, i)}= \frac{r}{ri+N-i} \cdot \frac{N-i}{N-1}+\frac{r(i-1)}{ri+N-i} \cdot \frac{N-i}{2(N-2)}.\tag{S47}
\end{equation}
If any other event occurs, then the state remains unchanged. Therefore, we obtain
\begin{equation}
p_{(1, i-1)\to(1, i-1)}=1-p_{(1, i-1)\to(0, i-1)}-p_{(1, i-1)\to(1, i-2)}-p_{(1, i-1)\to(1, i)}.\tag{S48}
\end{equation}
We remind that $\tilde{x}_{(i_1,i_2)}$ represents the probability that A fixates starting with state $(i_1, i_2)$. 
We obtain
\begin{linenomath}
\begin{align}
\tilde{x}_{(1, i-1)} &= p_{(1, i-1) \to (0, i-1)} \tilde{x}_{(0, i-1)} + p_{(1, i-1) \to (1, i-2)} \tilde{x}_{(1, i-2)} + 
p_{(1, i-1) \to (1, i)} \tilde{x}_{(1, i)} + p_{(1, i-1) \to (1, i-1)} \tilde{x}_{(1, i-1)}.\tag{S49}
\label{eq:wstar-rho1}
\end{align}
\end{linenomath}

Now we assume that the current state is $(i_1, i_2) = (0, i)$ with $i \in \{1, 2, \ldots, N-1 \}$. There are four types of events that can occur in the next time step.
In the first type of event, a leaf node of type A is selected as parent with probability $ri/(ri+N-i)$. Then, the edge between the parent and the hub node is used for reproduction with probability $1/2$. The state after this event is $(i_1, i_2) = (1, i)$. Therefore, we obtain
\begin{equation}
p_{(0, i) \to (1, i)} = \frac{ri}{ri+N-i} \cdot \frac{1}{2}.\tag{S50}
\end{equation}
In the second type of event, a leaf node of type A is selected as parent with probability $ri/(ri+N-i)$. Then, the edge between the parent and a leaf node of type B is used for reproduction with probability $(N-i-1)/2(N-2)$. The state after this event is $(0, i+1)$. Therefore, we obtain
\begin{equation}
p_{(0, i) \to (0, i+1)} = \frac{ri}{ri+N-i} \cdot \frac{N-i-1}{2(N-2)}.\tag{S51}
\end{equation} 
In the third type of event, the hub node, which is of type B, is selected as parent with probability $1/(ri+N-i)$. Then, the edge between the parent and a leaf node of type A is used for reproduction with probability $i/(N-1)$. Alternatively, a leaf node of type B is selected as parent with probability $(N-i-1)/(ri+N-i)$. Then, the edge between the parent and a leaf node of type A is used for reproduction with probability $i/2(N-2)$. In both cases, the state after the event is $(0, i-1)$. Therefore, we obtain
\begin{equation}
p_{(0, i) \to (0, i-1)}= \frac{1}{ri+N-i} \cdot \frac{i}{N-1}+\frac{N-i-1}{ri+N-i} \cdot \frac{i}{2(N-2)}.\tag{S52}
\end{equation}
If any other event occurs, then the state remains unchanged. Therefore, we obtain
\begin{equation}
p_{(0, i)\to(0, i)}=1-p_{(0, i)\to(1, i)}-p_{(0, i)\to(0, i+1)}-p_{(0, i)\to(0, i-1)}.\tag{S53}
\end{equation}
Using these transition probabilities, we obtain
\begin{linenomath}
\begin{align}
\tilde{x}_{(0, i)} &= p_{(0, i) \to (1, i)} \tilde{x}_{(1, i)} + p_{(0, i) \to (0, i+1)} \tilde{x}_{(0, i+1)} + 
p_{(0, i) \to (0, i-1)} \tilde{x}_{(0, i-1)} + p_{(0, i) \to (0, i)} \tilde{x}_{(0, i)}.\tag{S54}
\label{eq:wstar-rho2}
\end{align}
\end{linenomath}
Equation~(37) also holds true for the one-mode projection. We use the scipy DGESV algorithm to numerically solve Eq.~(37) to obtain $\tilde{x}_{(0, i)}$ and $\tilde{x}_{(1, i-1)}$. Then, we obtain $x_i$ using Eq.~(39).

We computed $x_1$ by numerically solving Eq.~(37) for $N=4, 5, 20$, and $200$. Fig~\ref{fig:star_omp} compares the obtained $x_1$ values with those for the Moran process and the star 3-uniform hypergraph under model 1. We find that the one-mode projection of the star 3-uniform hypergraph is a weak amplifier of selection. Note that, for $N=4$, $5$, and $200$, the results for the one-mode projection almost overlap those for the Moran process such that the orange lines are hidden behind the black lines in Fig~\ref{fig:star_omp}A, \ref{fig:star_omp}B, and \ref{fig:star_omp}D.

\renewcommand{\thefigure}{B}
\floatsetup[figure]{style=plain,subcapbesideposition=top}
\captionsetup{font={small,rm}} 
\captionsetup{labelfont=bf}
\begin{figure}[H]
  \centering
  \includegraphics[width=0.9\linewidth]{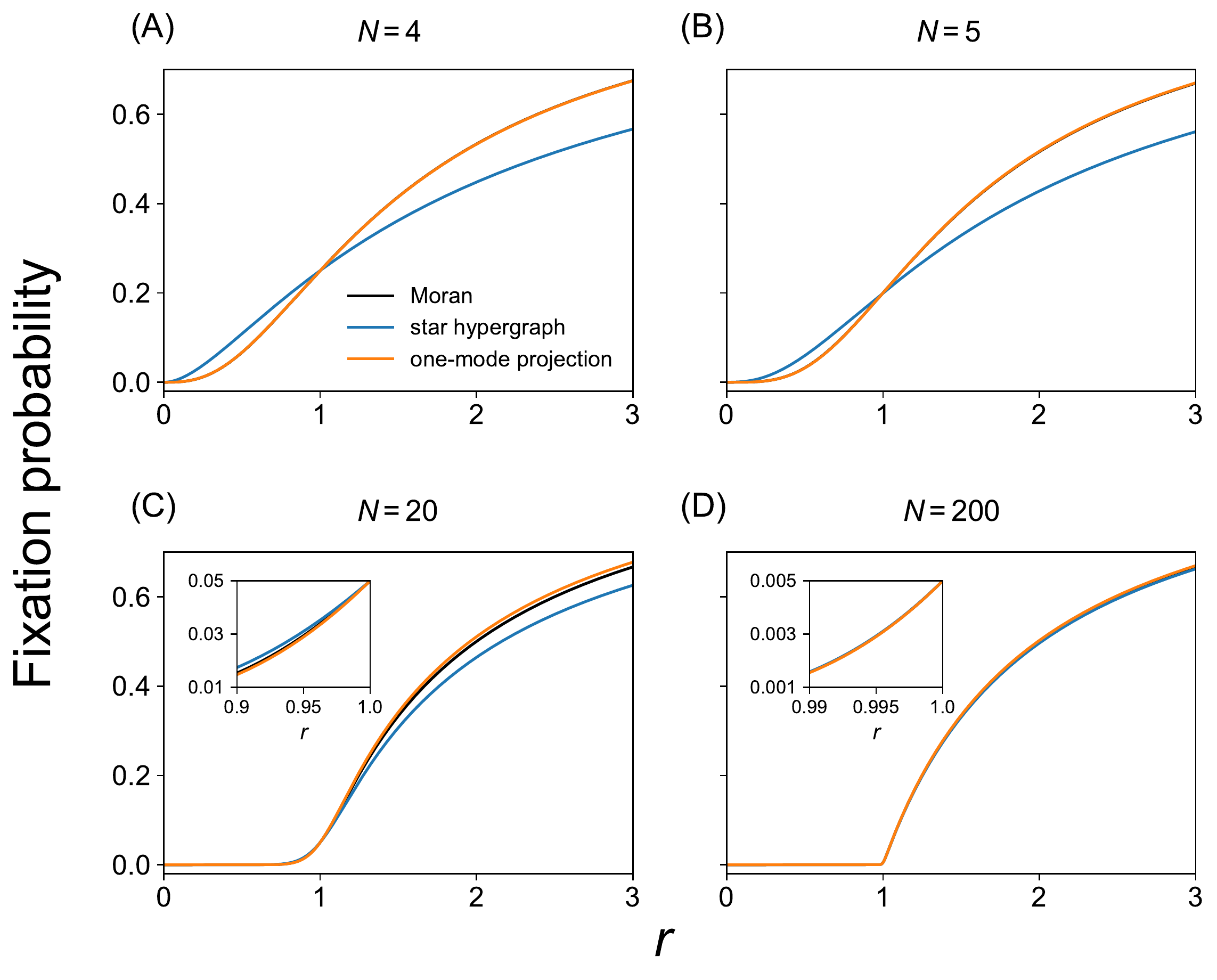}
   \caption{Fixation probability for the weighted one-mode projection of star 3-uniform hypergraphs. We compare it with the fixation probability for the Moran process and star 3-uniform hypergraphs. (A) $N=4$. (B) $N=5$. (C) $N=20$. (D) $N=200$. The insets in (C) and (D) magnify the results for $r$ values smaller than and close to $r=1$. In the inset in (D), the results for the star 3-uniform hypergraph (shown by the blue line) and the Moran process (shown by the black line) are close to that for the one-mode projection (shown by the orange line) such that the blue and the black lines are almost hidden behind the orange line. In (A), (B), and the main panel of (D), the results for the Moran process (shown by the black lines) are not identical but close to those for the one-mode projection (shown by the orange lines) such that the former are hidden behind the latter.}
   \label{fig:star_omp}
\end{figure}

\section*{Text J. Derivation of the fixation probability for the cyclic 3-uniform hypergraph under model 2\label{sec:cyclic-model2-derivation}}

In this section, we derive the fixation probability for the cyclic 3-uniform hypergraph under model 2. We assume that there are initially just two mutants that are uniformly distributed.

The fixation of type A can occur only when the two nodes that are initially of type A share at least one hyperedge. Once such a hyperedge is selected for reproduction, all the nodes of type A are consecutive along the cycle without being interrupted by nodes of type B. Note that the two nodes that initially have type A may be next to each other on the cycle already in the initial condition. Therefore, to calculate the fixation probability on the cyclic 3-uniform hypergraph, it suffices to track the number of consecutive nodes having type A, which we denote by $i$. For $N\geq 5$, the initial condition is either of the following three types.

\subsection*{First type of initial condition}

In the first type of initial condition, the two nodes of type A do not share any hyperedge (see Fig~\ref{fig:m2_cyclic}A for a schematic), which occurs with probability
$N(N-5)/\left[2\binom{N}{2}\right]$. The probability that type A fixates under this initial condition, denoted by $x_2^\prime$, is given by
\begin{equation}
x_2^{\prime} = 0.\tag{S55}
\label{eq:cyclic-model2-x_2'=0}
\end{equation}

\subsection*{Second type of initial condition}

In the second type of initial condition, the two nodes of type A share two hyperedges, i.e., these two nodes are next to each other on the cycle (see Fig~\ref{fig:m2_cyclic}B). This initial condition occurs with probability $N/\binom{N}{2}$. Let $x_2^{\prime\prime}$ be the fixation probability for type A under this initial condition. In general, let $x_i^{\prime\prime}$ with $i\in \{1, 2, \ldots, N-1\}$ be the fixation probability for type A when there are $i$ consecutive nodes of type A and all the other nodes are of type B. We calculate $x_2^{\prime\prime}$ by tracking the number of consecutive nodes with type A, i.e., $i$, as follows.

\textbf{\underline{Move of the state from $i$ to $i-1$:}}
If $i\in \{2, \ldots, N-2\}$, there are three types of events that can occur next. Without loss of generality, we assume that the $\ell$th to the ($\ell+i-1$)th nodes are of type A and that all the other nodes are of type B (see Fig~3C).
In the first type of event, the state moves from $i$ to $i-1$. If $i\le N-4$, either $v_{\ell-2}$, $v_{\ell-1}$, $v_{\ell+i}$, and $v_{\ell+i+1}$, which is of type B, is selected as parent with probability $4/(ri+N-i)$. Then, the hyperedge that contains the parent node, a node of type B, and a node of type A is used for reproduction with probability $1/3$. In this case, $i$ decreases by one. For example, if $v_{\ell-2}$ is selected as parent and hyperedge $\{ \ell-2, \ell-1, \ell \}$ is used for reproduction, then the state moves from $i$ to $i-1$. Therefore, we obtain
\begin{equation}
p_{i,i-1}=\frac{4}{ri+N-i}\cdot\frac{1}{3}.\tag{S56}
\label{eq:model2-cyclic-decrease1}
\end{equation}
\renewcommand{\thefigure}{C}
\floatsetup[figure]{style=plain,subcapbesideposition=top}
\captionsetup{font={small,rm}} 
\captionsetup{labelfont=bf}
\begin{figure}[H]
  \centering
  \includegraphics[width=1.0\linewidth]{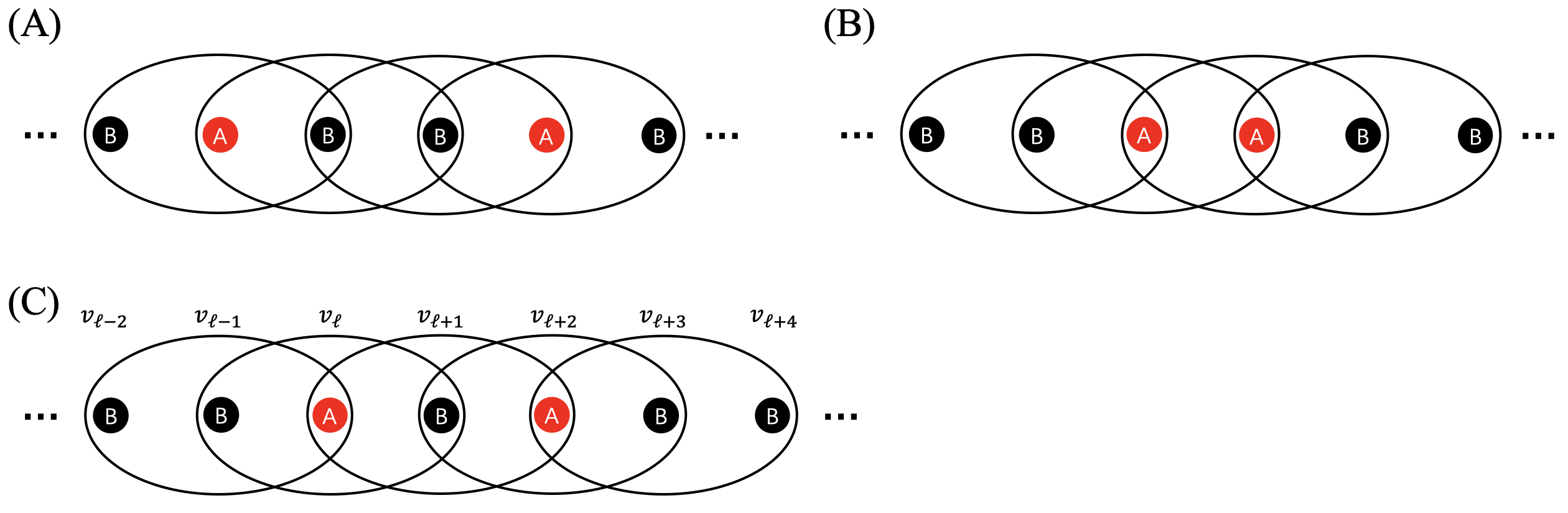}
\caption{Initial position of the two nodes of type A on the cyclic 3-uniform hypergraph. (A) The two nodes of type A do not share any hyperedge. (B) The two nodes of type A share two hyperedges. (C) The two nodes of type A share one hyperedge.}   
   \label{fig:m2_cyclic}
\end{figure}
\noindent If $i=N-3$, there are two sequences of events through which the state decreases from $i$ to $i-1$. In the first sequence, either $v_{\ell-1}$ or $v_{\ell-3}$ is selected as parent, which occurs with probability $2/(rN-3r+3)$. Then, the hyperedge containing two nodes of type B (i.e., hyperedge $\{\ell-2, \ell-1, \ell\}$ if $v_{\ell-1}$ is the parent and hyperedge $\{\ell-4, \ell-3, \ell-2\}$ if $v_{\ell-3}$ is the parent) is used for reproduction, which occurs with probability $1/3$. In the second sequence, $v_{\ell-2}$ is selected as parent with probability $1/(rN-3r+3)$. Then, hyperedge $\{\ell-2, \ell-1, \ell\}$ or $\{\ell-4, \ell-3, \ell-2\}$ is used for reproduction, which occurs with probability $2/3$. By summing up these probabilities, we obtain Eq.~\eqref{eq:model2-cyclic-decrease1}. If $i=N-2$, either of the two nodes of type B, i.e., $v_{\ell-1}$ or $v_{\ell-2}$, must be selected as parent for the state to move from $i$ to $i-1$. This event occurs with probability $2/(rN-2r+2)$. Then, either hyperedge $\{\ell-2, \ell-1, \ell\}$ or $\{\ell-3, \ell-2, \ell-1\}$ must be used for reproduction, which occurs with probability $2/3$. The product of these two probabilities coincides with Eq.~\eqref{eq:model2-cyclic-decrease1}. Therefore, Eq.~\eqref{eq:model2-cyclic-decrease1} holds true for any $i \in \{2, \ldots, N-2\}$.

\textbf{\underline{Move of the state from $i$ to $i+1$:}}
In the second type of event, the state moves from $i$ to $i+1$.
If $i\ge 4$, either $v_\ell$, $v_{\ell+1}$, $v_{\ell+i-2}$, or $v_{\ell+i-1}$, which is of type A, has to be selected as parent with probability $4r/(ri+N-i)$. Then, the hyperedge that contains the parent node, a node of type A, and a node of type B has to be used for reproduction, which occurs with probability $1/3$. For example, if $v_\ell$ is selected as parent and hyperedge $\{ \ell-1, \ell, \ell+1 \}$ is used for reproduction, then the state moves from $i$ to $i+1$. Therefore, we obtain
\begin{equation}
p_{i,i+1}=\frac{4r}{ri+N-i}\cdot\frac{1}{3}.\tag{S57}
\label{eq:model2-cyclic-increase1}
\end{equation}
If $i=3$, there are two sequences of events through which the state increases from $i$ to $i+1$. In the first sequence, either $v_{\ell}$ or $v_{\ell+2}$ is selected as parent, which occurs with probability $2r/(3r+N-3)$. Then, the hyperedge containing two nodes of type A (i.e., hyperedge \{${\ell-1}$, ${\ell}$, ${\ell+1}$\} if $v_{\ell}$ is the parent and hyperedge \{${\ell+1}$, ${\ell+2}$, ${\ell+3}$\} if $v_{\ell+2}$ is the parent) is used for reproduction, which occurs with probability $1/3$. In the second sequence, $v_{\ell+1}$ is selected as parent with probability $r/(3r+N-3)$. Then, hyperedge \{${\ell-1}$, ${\ell}$, ${\ell+1}$\} or
\{${\ell+1}$, ${\ell+2}$, ${\ell+3}$\} is used for reproduction, which occurs with probability $2/3$. If we sum these probabilities, we obtain Eq.~\eqref{eq:model2-cyclic-increase1}. If $i=2$, either of the two nodes of type A, i.e., $v_{\ell}$ or $v_{\ell+1}$, must be selected as parent for the state to move from $i$ to $i+1$. This event occurs with probability $2r/(2r+N-2)$. Then, either hyperedge \{${\ell-1}$, ${\ell}$, ${\ell+1}$\} or  \{${\ell}$, ${\ell+1}$, ${\ell+2}$\} must be used for reproduction, which occurs with probability $2/3$. The product of these two probabilities coincides with Eq.~\eqref{eq:model2-cyclic-increase1}. Therefore, Eq.~\eqref{eq:model2-cyclic-increase1} holds true for any $i\in \{2, \ldots, N-2\}$. 

\textbf{\underline{No move of the state from $i$:}}
Because $i$ remains unchanged if any other event occurs, we obtain 
\begin{equation}
p_{i,i}=1-p_{i,i-1}-p_{i,i+1}.\tag{S58}
\label{eq:model2-p_ii}
\end{equation}

\textbf{\underline{Derivation of $x_2^{\prime\prime}$:}}
Therefore, the fixation probability of type A starting from $i$ consecutive nodes of type A, i.e., $x_i^{\prime\prime}$, satisfies
\begin{linenomath}
\begin{align}
x_0 &= x_1 = 0,\tag{S59}
\label{eq:x0-m2-cyclic}\\
x_i^{\prime\prime} &= p_{i,i-1}x_{i-1}^{\prime\prime}+p_{i,i}x_{i}^{\prime\prime}+p_{i,i+1}x_{i+1}^{\prime\prime}, \quad i\in \{2, \ldots, N-2\},\tag{S60}
\label{eq:recursive-m2-cyclic}\\
x_{N-1} &= x_N = 1.\tag{S61}
\label{xn-1-m2-cyclic}
\end{align}
\end{linenomath}
Note that $x_1^{\prime\prime} = x_1$ and $x_{N-1}^{\prime\prime} = x_{N-1}$.
Similar to the analysis of the fixation probability for the complete 3-uniform hypergraph, we set
\begin{equation}
\overline{y}_i \equiv x_i^{\prime\prime}-x_{i-1}^{\prime\prime}, \quad i\in \{2, \ldots, N-1\}.\tag{S62}
\label{eq:model2-overliney-def}
\end{equation}
Note that $\sum_{i=2}^{N-1} \overline{y}_i = x_{N-1}^{\prime\prime}-x_1^{\prime\prime}=1$. Let
\begin{equation}
\overline{\gamma}_i=p_{i,i-1}/p_{i,i+1}.\tag{S63}
\label{eq:model2-overlinegamma-def}
\end{equation}
By combining Eqs.~\eqref{eq:model2-p_ii}, \eqref{eq:recursive-m2-cyclic}, \eqref{eq:model2-overliney-def}, and \eqref{eq:model2-overlinegamma-def}, 
we obtain
\begin{equation}
\overline{y}_{i+1} = \overline{y}_i \overline{\gamma}_i,\tag{S64}
\end{equation}
which leads to
\begin{equation}
\overline{y}_i = \overline{y}_2 \prod_{k=2}^{i-1} \overline{\gamma}_k = x_2^{\prime\prime} \prod_{k=2}^{i-1} \overline{\gamma}_k.\tag{S65}
\label{eq:overeliney-in-overelinegamma}
\end{equation}
Using Eq.~\eqref{eq:overeliney-in-overelinegamma}, we obtain
\begin{linenomath}
\begin{align}
1 = \sum_{i=2}^{N-1} \overline{y}_i &=x_2^{\prime\prime}\left[1+\overline{\gamma}_2+\overline{\gamma}_2\overline{\gamma}_3 + \cdots + \prod_{k=2}^{N-2} \overline{\gamma}_k\right] \notag\\
&=x_2^{\prime\prime}\left[1+r^{-1}+r^{-2}+\cdots+r^{-(N-3)}\right] \notag\\
&=x_2^{\prime\prime}\frac{1-r^{-(N-2)}}{1-r^{-1}}.\tag{S66}
\label{eq:model2-cycle-x2-infer}
\end{align}
\end{linenomath}
Therefore, we obtain
\begin{equation}
x_2^{\prime\prime}=\frac{1-r^{-1}}{1-r^{-(N-2)}}.\tag{S67}
\label{eq:model2-cycle-x2-case2}
\end{equation}

\subsection*{Third type of initial condition}

In the third type of initial condition, the two nodes of type A share one hyperedge, implying that there is a node of type B between the two nodes of type A.
Without loss of generality, we assume that the $\ell$th and the $(\ell+2)$th nodes are of type A and that all the other nodes are of type B (see Fig~\ref{fig:m2_cyclic}C). This initial condition, which we denote by $2^{*}$, occurs with probability $N/\binom{N}{2}$. Now we calculate the fixation probability for type A starting from state $2^{*}$, which we denote by $x_2^{'''}$. To ease the discussion, in the remainder of this section, we denote by $i$ the state in which consecutive $i$ nodes on the cycle are of type A and the other $N-i$ nodes are of type B.

\textbf{\underline{Move of the state from $2^{*}$ to $i=1$:}}
If $N\ge 7$, the state moves from $2^{*}$ to $i=1$ in one time step if either of the following two types of events occurs. In the first type of event, either node $v_{\ell-2}$ or $v_{\ell+4}$, which is of type B, is selected as parent. This event occurs with probability $2/(2r+N-2)$. Then, the hyperedge that contains the parent, a node of type B, and a node of type A, is used for reproduction, which occurs with probability $1/3$. For example, if $v_{\ell-2}$ is the parent and hyperedge $\{\ell-2, \ell-1, \ell\}$ is used for reproduction, then the state moves from $2^{*}$ to $i=1$. In the second type of event, one of the nodes $v_{\ell-1}$, $v_{\ell+1}$, and $v_{\ell+3}$, which is of type B, is selected as parent, which occurs with probability $3/(2r+N-2)$. Then, one of the two hyperedges that contains the parent, a node of type B, and a node of type A, is used for reproduction, which occurs with probability $2/3$. For example, if $v_{\ell-1}$ is selected as parent and hyperedge $\{ \ell-2, \ell-1, \ell \}$ or $\{ \ell-1, \ell, \ell+1 \}$ is used for reproduction, then the state moves from $2^{*}$ to $i=1$. By summing up these probabilities, we obtain the probability that the state moves from $2^{*}$ to $i=1$ as
\begin{equation}
p_{2^{*},1}=\frac{1}{2r+N-2}\cdot\frac{8}{3}.\tag{S68}
\label{eq:model2-cycle-x2-case3-decrease1}
\end{equation}

If $N=6$, the state moves from $2^*$ to $i=1$ in one time step if the following event occurs. Either $v_{\ell-2}$, $v_{\ell-1}$, $v_{\ell+1}$, and $v_{\ell+3}$, which is of type B, is selected as parent with probability $4/(2r+4)$. Then, one of the two hyperedges that contains the parent, a node of type B, and a node of type A, is used for reproduction, which occurs with probability $2/3$. For example, if $v_{\ell-2}$ is selected as parent and hyperedge $\{ \ell-4, \ell-3, \ell-2 \}$ or $\{ \ell-2, \ell-1, \ell \}$ is used for reproduction, then the state moves from $2^{*}$ to $i=1$. The product of these two probabilities coincides with Eq.~\eqref{eq:model2-cycle-x2-case3-decrease1}. 

If $N=5$, the state moves from $2^{*}$ to $i=1$ in one time step if either of the following two types of events occurs. In the first type of event, either node $v_{\ell-1}$ or $v_{\ell+3}$, which is of type B, is selected as parent with probability $2/(2r+3)$. Then, the hyperedge that contains the parent, a node of type B, and a node of type A, is used for reproduction, which occurs with probability $1$. For example, if $v_{\ell-1}$ is the parent and any hyperedge that contains $v_{\ell-1}$ is used for reproduction, then the state moves from $2^*$ to $i=1$. In the second type of event, the node $v_{\ell+1}$, which is of type B, is selected as parent with probability $1/(2r+3)$. Then, one of the two hyperedges $\{\ell-1, \ell, \ell+1\}$ or $\{\ell+1, \ell+2, \ell+3\}$ is used for reproduction, which occurs with probability $2/3$. By summing up these probabilities, we obtain Eq.~\eqref{eq:model2-cycle-x2-case3-decrease1}. Therefore, Eq.~\eqref{eq:model2-cycle-x2-case3-decrease1} holds true for any $N\ge 5$.

\textbf{\underline{Move of the state from $2^{*}$ to $i=3$:}}
The state moves from $2^{*}$ to $i=3$ if either $v_{\ell}$ or $v_{\ell+2}$, which is of type A, is selected as parent with probability $2r/(2r+N-2)$, and then, the hyperedge that contains $v_{\ell}$, $v_{\ell+1}$, and $v_{\ell+2}$ is used for reproduction with probability $1/3$. Therefore, we obtain
\begin{equation}
p_{2^{*},3}=\frac{2r}{2r+N-2}\cdot\frac{1}{3}.\tag{S69}
\end{equation}

\textbf{\underline{No move of the state from $2^{*}$:}}
If any other event occurs at state $2^{*}$, the state remains unchanged. Therefore, we obtain
\begin{linenomath}
\begin{align}
p_{2^{*},2^{*}}&=1-p_{2^{*},1}-p_{2^{*},3}=\frac{4r+3N-14}{3(2r+N-2)},\tag{S70}\\
p_{2^{*},j} &= 0 \text{ if } j \neq 1, 2^{*}, 3.\tag{S71}
\end{align}
\end{linenomath}  

\textbf{\underline{Derivation of $x_2^{\prime\prime\prime}$:}}
If the state moves from $2^{*}$ to either $i=1$ or $i=3$, all the nodes of type A are consecutively numbered without being interrupted by nodes of type B afterwards. Therefore, we obtain
\begin{equation}
x_2^{'''}=p_{2^*,1}x_1+p_{2^*,2^*}x_2^{'''}+p_{2^*,3}x_3^{\prime\prime}.\tag{S72}
\label{x2-triprime}
\end{equation}
By substituting Eqs.~\eqref{eq:x0-m2-cyclic} and \eqref{eq:model2-cycle-x2-case2} in Eq.~\eqref{eq:recursive-m2-cyclic} for $i=2$, we obtain
\begin{equation}
x_3^{\prime\prime}=\frac{1-r^{-2}}{1-r^{-(N-2)}}.\tag{S73}
\label{eq:model2-cycle-x3-case2}
\end{equation}
By substituting Eqs.~\eqref{eq:x0-m2-cyclic} and \eqref{eq:model2-cycle-x3-case2} in Eq.~\eqref{x2-triprime}, we obtain
\begin{equation}
x_2^{'''}=\frac{r-r^{-1}}{(r+4)\left[1-r^{-(N-2)}\right]}.\tag{S74}
\label{eq:model2-cycle-x2-case3}
\end{equation}

\subsection*{Weighted sum to obtain $x_2$}

By combining Eqs.~\eqref{eq:cyclic-model2-x_2'=0}, \eqref{eq:model2-cycle-x2-case2}, and \eqref{eq:model2-cycle-x2-case3} with the respective probability,
we obtain the fixation probability for type A when there are initially two uniformly randomly distributed mutants, given in
Eq.~(70), as follows:
\begin{linenomath}
\begin{align}
x_2&= \frac{N-5}{N-1}x_2^{\prime} + \frac{2}{N-1}x_2^{''}+\frac{2}{N-1}x_2^{'''} \notag\\
&=\frac{2}{N-1}\left\{\frac{1-r^{-1}}{1-r^{-(N-2)}}+\frac{r-r^{-1}}{(r+4)\left[1-r^{-(N-2)}\right]}\right\},\tag{S75}
\label{eq:model2-cyclic-x2-appendix}
\end{align}
\end{linenomath}
where $N \ge 5$.

\subsection*{Derivation of $x_2$ for $N=4$}

For $N=4$, the first type of initial condition occurs with probability $0$. 

The second type of initial condition occurs with the same probability as in the case of $N\ge 5$, i.e., with probability $4/\binom{4}{2}=2/3$. In this case, the state moves from $2$ to $1$ if the following event occurs. Either $v_{\ell-1}$ or $v_{\ell+2}$, which is of type B, is selected as parent with probability $2/(2r+2)$. Then, either hyperedge $\{\ell+2, \ell-1, \ell\}$ or $\{\ell+1, \ell+2, \ell-1\}$ must be used for reproduction, which occurs with probability $2/3$. 
Note that the node indices $\ell-1$ and $\ell$ are equivalent to $\ell+3$ and $\ell+4$ because we interpret the node index with modulo $N$.
The product of these two probabilities coincides with Eq.~\eqref{eq:model2-cyclic-decrease1}. Therefore, Eq.~\eqref{eq:model2-cyclic-decrease1} also holds true for $N=4$. Similarly, Eq.~\eqref{eq:model2-cyclic-increase1} holds true for $N=4$. Therefore, using Eq.~\eqref{eq:model2-cycle-x2-infer}, we obtain
\begin{equation}
x_2^{\prime\prime}=\frac{r}{r+1}.\tag{S76}
\label{eq:model2-cycle-x2-case2-N4}
\end{equation}

Under the third type of initial condition, which occurs with probability $2/\binom{4}{2}=1/3$, the state moves from $2^*$ to $1$ if the following event occurs. Either $v_{\ell-1}$ or $v_{\ell+1}$, which is of type B, is selected as parent with probability $2/(2r+2)$. Then, either hyperedge $\{\ell-1, \ell, \ell+1\}$ or $\{\ell+1, \ell+2, \ell-1\}$ must be used for reproduction, which occurs with probability $2/3$.  Therefore, we obtain
\begin{equation}
p_{2^{*},1}=\frac{2}{3r+3}.\tag{S77}
\label{eq:model2-cycle-x2-case3-decrease1-N4}
\end{equation}
The state moves from $2^*$ to $3$ if the following event occurs. Either $v_{\ell}$ or $v_{\ell+2}$, which is of type A, is selected as parent with probability $2r/(2r+2)$. Then, either hyperedge $\{\ell, \ell+1, \ell+2\}$ or $\{\ell+2, \ell-1, \ell\}$ must be used for reproduction, which occurs with probability $2/3$. Therefore, we obtain
\begin{equation}
p_{2^{*},3}=\frac{2r}{3r+3}.\tag{S78}
\label{eq:model2-cycle-x2-case3-increase1-N4}
\end{equation}
If any other event occurs at state $2^{*}$, the state remains unchanged. Therefore, we obtain
\begin{linenomath}
\begin{align}
p_{2^{*},2^{*}}&=1-p_{2^{*},1}-p_{2^{*},3}=\frac{r+1}{3r+3},\tag{S79}\\
p_{2^{*},j} &= 0 \text{ if } j \neq 1, 2^{*}, 3.\tag{S80}
\end{align}
\end{linenomath}  
By substituting Eqs.~\eqref{eq:x0-m2-cyclic} and \eqref{xn-1-m2-cyclic} in Eq.~\eqref{x2-triprime}, we obtain
\begin{equation}
x_2^{'''}=\frac{r}{r+1}.\tag{S81}
\label{eq:model2-cycle-x2-case3-N4}
\end{equation}

Therefore, we obtain
\begin{equation}
x_2=\frac{2}{3}x_2^{''}+\frac{1}{3}x_2^{'''}=\frac{r}{r+1}\tag{S82}
\end{equation}
for $N=4$.

\section*{Text K. Derivation of the fixation probability for the star 3-uniform hypergraph under model 2\label{sub:star-model2-derivation}}

We derive the fixation probability for the star 3-uniform hypergraph under model 2 in this section. We use same notations as those in section~3.1.3.

Assume that the current state is $(i_1, i_2) = (1, i-1)$ with $i-1 \in \{0, 1, \ldots, N-2 \}$. There are three types of events that can occur in the next time step.
In the first type of event, a leaf node of type B is selected as parent with probability $(N-i)/(ri+N-i)$. Then, a hyperedge that contains the parent, the hub node, and a different leaf node of type B, is used for reproduction with probability $(N-i-1)/(N-2)$. The state after this entire event is $(i_1, i_2) = (0, i-1)$. Therefore, we obtain
\begin{equation}
p_{(1, i-1) \to (0, i-1)} = \frac{N-i}{ri+N-i} \cdot \frac{N-i-1}{N-2}.\tag{S83}
\end{equation}
In the second type of event, the hub node, which is of type A, is selected as parent with probability $r/(ri+N-i)$. Then, a hyperedge that contains the parent, a leaf node of type A, and a leaf node of type B, is used for reproduction with probability $(i-1)(N-i)/\binom{N-1}{2}$. Alternatively, a leaf node of type A is selected as parent with probability $r(i-1)/(ri+N-i)$. Then, a hyperedge that contains the parent, the hub node, and a leaf node of type B, is used for reproduction with probability $(N-i)/(N-2)$. In both cases, the state after the event is $(1, i)$. Therefore, we obtain
\begin{equation}
p_{(1, i-1) \to (1, i)} = \frac{r}{ri+N-i} \cdot \frac{(i-1)(N-i)}{\binom{N-1}{2}}+\frac{r(i-1)}{ri+N-i} \cdot \frac{N-i}{N-2}.\tag{S84}
\end{equation}
If any other event occurs, then the state remains unchanged. Therefore, we obtain
\begin{equation}
p_{(1, i-1) \to (1, i-1)} = 1 - p_{(1, i-1) \to (0, i-1)}- p_{(1, i-1) \to (1, i)}.\tag{S85}
\end{equation}
We remind that $\tilde{x}_{(i_1,i_2)}$ is the probability that type A fixates when the initial state is $(i_1, i_2)$.
We obtain
\begin{linenomath}
\begin{align}
\tilde{x}_{(1, i-1)} &= p_{(1, i-1) \to (0, i-1)} \tilde{x}_{(0, i-1)} + p_{(1, i-1) \to (1, i)} \tilde{x}_{(1, i)} + p_{(1, i-1) \to (1, i-1)} \tilde{x}_{(1, i-1)}.\tag{S86}
\label{eq:m2star-rho1}
\end{align}
\end{linenomath}

Now we assume that the current state is $(i_1, i_2) = (0, i)$ with $i \in \{1, \ldots, N-1 \}$. Then, there are three types of events that can occur in the next time step. In the first type of event, a leaf node of type A is selected as parent with probability $ri/(ri+N-i)$. Then, a hyperedge that contains the parent, the hub node, and a different leaf node of type A, is used for reproduction with probability $(i-1)/(N-2)$. The state after this event is $(i_1,i_2)=(1, i)$. Therefore, we obtain
\begin{equation}
p_{(0, i) \to (1, i)} = \frac{ri}{ri+N-i} \cdot \frac{i-1}{N-2}.\tag{S87}
\end{equation}
In the second type of event, the hub node, which is of type B, is selected as parent with probability $1/(ri+N-i)$. Then, a hyperedge that contains the parent, a leaf node of type A, and a leaf node of type B, is used for reproduction with probability $i(N-i-1)/\binom{N-1}{2}$. Alternatively, a leaf node of type B is selected as parent with probability $(N-i-1)/(ri+N-i)$. Then, a hyperedge that contains the parent, the hub node, and a leaf node of type A, is used for reproduction with probability $i/(N-2)$. In both cases, the state after the event is $(0, i-1)$. Therefore, we obtain
\begin{equation}
p_{(0, i) \to (0, i-1)} = \frac{1}{ri+N-i} \cdot \frac{i(N-i-1)}{\binom{N-1}{2}}+\frac{N-i-1}{ri+N-i} \cdot \frac{i}{N-2}.\tag{S88}
\end{equation}
If any other event occurs, then the state remains unchanged. Therefore, we obtain
\begin{linenomath}
\begin{align}
p_{(0, i) \to (0, i)} = 1 - p_{(0, i) \to (1, i)} - p_{(0, i) \to (0, i-1)}.\tag{S89}
\end{align}
\end{linenomath}
Using these transition probabilities, we obtain
\begin{linenomath}
\begin{align}
\tilde{x}_{(0, i)} &= p_{(0, i) \to (1, i)} \tilde{x}_{(1, i)} + p_{(0, i) \to (0, i-1)} \tilde{x}_{(0, i-1)} + p_{(0, i) \to (0, i)} \tilde{x}_{(0, i)}.\tag{S90}
\label{eq:m2star-rho2}
\end{align}
\end{linenomath}

Equations~\eqref{eq:m2star-rho1} and \eqref{eq:m2star-rho2} lead to
\begin{equation}
\tilde{x}_{(1,i)} = \alpha_i\tilde{x}_{(1,i-1)}+(1-\alpha_i)\tilde{x}_{(0,i-1)}\tag{S91}
\label{alpha}
\end{equation}
and
\begin{equation}
\tilde{x}_{(0,i)} = \beta_i\tilde{x}_{(1,i)}+(1-\beta_i)\tilde{x}_{(0,i-1)},\tag{S92}
\label{beta}
\end{equation}
respectively, where
\begin{linenomath}
\begin{align}
\alpha_i&=1+\frac{(N-i-1)(N-1)}{r(i-1)(N+1)},\tag{S93}\\
\beta_i&=\frac{r(i-1)(N-1)}{r(i-1)(N-1)+(N-i-1)(N+1)}.\tag{S94}
\end{align}
\end{linenomath}
We rewrite Eqs.~\eqref{alpha} and \eqref{beta} as
\begin{equation}\label{recursion}
\bm{\varrho_i}=A_i\bm{\varrho_{i-1}},\tag{S95}
\end{equation}
where $\bm{\varrho_i}=(\tilde{x}_{(1,i)}, \tilde{x}_{(0,i)})^{\top}$, and $A_i$ is the $2\times 2$ matrix given by
\begin{equation}
A_i=
\begin{pmatrix}
\alpha_i & 1-\alpha_i \\
\alpha_i\beta_i & 1-\alpha_i\beta_i 
\end{pmatrix}.\tag{S96}
\end{equation}
Equation~\eqref{recursion} yields
\begin{equation}
\bm{\varrho_i}=A_iA_{i-1}\cdots A_2\bm{\varrho_{1}}.\tag{S97}
\end{equation}
Therefore,
\begin{equation}\label{recursion-to-N}
\begin{pmatrix}
1\\
1
\end{pmatrix}=
\bm{\varrho_{N-1}}=A_{N-1}A_{N-2}\cdots A_2\bm{\varrho_{1}}.\tag{S98}
\end{equation}
Equation~(39) also holds true for model 2. Therefore, we obtain $\tilde{x}_{(1,1)}$ from Eq.~\eqref{recursion-to-N}, $\tilde{x}_{(0,2)}$ from Eq.~\eqref{recursion}, and finally $x_2$ from Eq.~(39).

By analytically solving Eq.~\eqref{recursion-to-N} for $N=4$ and $N=5$, we obtain Eqs.~(55) and (56), respectively.

\end{document}